\documentclass[12pt]{article}
\usepackage{epsfig,psfrag,amsmath,amssymb,latexsym}
\usepackage{amscd}
\usepackage{color}
\usepackage{amsfonts}
\usepackage{graphicx}
\usepackage{wasysym}
\usepackage{mathrsfs}
\usepackage{verbatim}
\usepackage{hyperref}
\usepackage{cancel}
\usepackage[utf8]{inputenc}
\numberwithin{equation}{section}
\newtheorem{thm}{Theorem}[section]

\newtheorem{theorem}[thm]{Theorem}

\newtheorem{corollary}[thm]{Corollary}
\newtheorem{assumption}[thm]{Assumption}

\newtheorem{lemma}[thm]{Lemma}
\newtheorem{remark}[thm]{Remark}

\newenvironment{proof}[1][Proof]{\textbf{#1.} }{\ \rule{0.5em}{0.5em}}

\newcommand{\htwo}{H^{2|2}}
\newcommand{\eps}{{\varepsilon}}
\newcommand{\Id}{\operatorname{Id}}
\newcommand{\diag}{\operatorname{diag}}
\newcommand{\supp}{\operatorname{supp}}
\newcommand{\range}{\operatorname{range}}
\newcommand{\sk}[1]{\left\langle{#1}\right\rangle}
\newcommand{\cI}{\mathcal{I}}
\newcommand{\cR}{\mathcal{R}}

\newcommand{\W}{\mathcal{W}}
\newcommand{\E}{{\mathbb E}}
\newcommand{\R}{{\mathbb R}}  
\newcommand{\N}{{\mathbb N}}  
\newcommand{\Z}{{\mathbb Z}}  
\newcommand{\T}{{\mathcal{T}}}
\newcommand{\G}{{\mathcal{G}}}  
\newcommand{\arcosh}{\operatorname{arcosh}}
\newcommand{\body}{\operatorname{body}}
\newcommand{\cdrei}{c_3}
\newcommand{\cvier}{{c_4}}
\newcommand{\ceins}{{c_1}}
\newcommand{\czwei}{{c_2}}
\newcommand{\wnull}{{\overline W_0}}

\begin{document}
\begin{center}
  {\LARGE Fluctuations in the non-linear supersymmetric hyperbolic sigma model with hierarchical interactions}\\[3mm]
{\large Margherita Disertori\footnote{Institute for Applied Mathematics
\& Hausdorff Center for Mathematics, 
University of Bonn, \\
Endenicher Allee 60,
D-53115 Bonn, Germany.
E-mail: {disertori@iam.uni-bonn.de}}
\hspace{1cm} 
Franz Merkl \footnote{Mathematisches Institut, Ludwig-Maximilians-Universit{\"a}t  M{\"u}nchen,
Theresienstr.\ 39,
D-80333 Munich,
Germany.
E-mail: merkl@math.lmu.de
}
\hspace{1cm} 
Silke W.W.\ Rolles\footnote{
Department of Mathematics, CIT, 
Technische Universit{\"{a}}t M{\"{u}}nchen,
Boltzmannstr.\ 3,
D-85748 Garching bei M{\"{u}}nchen,
Germany.
E-mail: srolles@cit.tum.de}
\\[3mm]
{\small \today}}\\[3mm]
\end{center}

\begin{abstract}
  We consider a class of non-linear supersymmetric hyperbolic sigma models
  on a hierarchical lattice.
  We prove that the hyperbolic distance between the superspin at any vertex and a
  reference superspin has
  asymptotically small fluctuations for large enough interactions,
  uniformly in the size of the graph and in the choice of the 
  vertex. This can be viewed as a spontaneous breaking of the Lorentz boost symmetry.
  \footnote{MSC2020 subject classifications: Primary 82B20; Secondary 60G60.}
\end{abstract}

\section{Introduction and results}

We study a non-linear supersymmetric hyperbolic sigma model, also called the $\htwo$ model,
which was introduced by Zirnbauer in \cite{zirnbauer-91} in the context of random
operators for quantum diffusion. It is defined as a statistical mechanics-type model
where the spins take values on the hyperbolic supermanifold $\htwo$. 
This model is expected to share some features of the localization/delocalization
transition in random band matrices. 
Indeed, in dimensions $d\ge 3$, the model exhibits a phase transition between
a phase with long-range order, examined in \cite{disertori-spencer-zirnbauer2010}, 
and a disordered phase, studied in \cite{disertori-spencer2010}. 

In \cite{sabot-tarres2012}, Sabot and Tarr\`es showed that the $\htwo$ model is closely
linked to a memory-dependent stochastic process called the vertex-reinforced jump process. After a time change, the latter is a
mixture of Markov jump processes. The corresponding environment can be described
in terms of a marginal of the $\htwo$ model transformed to horospherical coordinates.
In \cite{merkl-rolles-tarres2019}, an interpretation of the real components of the
horospherical coordinates was given in terms of crossing numbers
and rescaled local times of the vertex-reinforced jump process.
Sabot and Zeng in \cite{sabot-zeng15} reformulated the model in terms of a random
Schrödinger operator in a potential with short-range dependence.
The classical Dynkin isomorphism theorem between Gaussian free fields and Markov processes
has been transferred to the $\htwo$ model and variants by Bauerschmidt, Helmuth, and
Swan in \cite{bauerschmidt-helmuth-swan-dynkin-isomorphism2019} and 
\cite{bauerschmidt-helmuth-swan-geometry-of-rw-isomorphism2021}.
In recent years, the $\htwo$ model and related models have been extensively studied, see
the survey \cite{bauerschmidt-helmuth-survey2021} and,  
e.g., \cite{bauerschmidt-crawford-helmuth-swan-spanning-forests2021},   
\cite{crawford2021}, \cite{disertori-merkl-rolles2019}, and 
\cite{disertori-rapenne-roja-molina-zeng2023}.
  
In \cite{disertori-merkl-rolles2023}, we proved transience for the vertex-reinforced
jump process on $\Z^d$, $d\ge 1$, with long-range weights which do not decay too fast.
Our proof required estimates of the $u$-marginal
of the $\htwo$ model in horospherical coordinates which were based on estimates in a
hierarchical model.
In the present paper, we study the $u$- and $s$-marginals jointly in the hierarchical model and show that they both concentrate
asymptotically close to the origin. This does not follow from \cite{disertori-merkl-rolles2023} because the method used there does not give any information on the
$s$-marginal. 
On the other hand, the known monotonicity results from \cite{poudevigne22}
apply only to functions of $u$, but not to functions of~$s$. Corresponding monotonicity
results for $s$ are not known to the best of our knowledge. Therefore,
we cannot transfer our joint results for $u$ and $s$ to $\Z^d$ with long-range
weights.

\subsection{The non-linear supersymmetric hyperbolic sigma model}
\label{sec:susy}

The supermanifold $\htwo$ is constructed from the linear space $\R^{3|2}\ni v=(x,y,z,\xi,\eta)$
with $x,y,z$ as even components and $\xi,\eta$ as odd components in a real Grassmann algebra. The space
$\R^{3|2}$ is 
endowed with the inner product $\sk{v,v'}=xx'+yy'-zz'+\xi\eta'-\eta\xi'$. 
Then, $\htwo$ is obtained by imposing the non-linear
constraints $\sk{v,v}=-1$ and $z>0$, which means that only the upper sheet
of the hyperboloid
$z=+\sqrt{1+x^2+y^2+2\xi\eta}$ is selected. More details can be found in 
\cite{phdthesis-swan2020}, \cite[Appendix]{disertori-merkl-rolles2020}, and 
\cite{disertori-spencer-zirnbauer2010}.
In many cases, it is more convenient to work with horospherical coordinates
$\R^{2|2}\ni(u,s,\overline \psi,\psi)\mapsto(x,y,z,\xi,\eta )\in\htwo$ 
defined by
\begin{align}
& x=\sinh u-\left(\frac12s^2+\overline\psi\psi\right)e^{u}, \quad
                                    y=se^{u}, \quad
                                    z=\cosh u+\left(\frac12s^2+\overline\psi\psi\right)e^{u},
  \nonumber\\
& \xi=e^{u}\overline\psi, \quad
\eta=e^{u}\psi.
\label{eq:change-of-coordinates2}
\end{align}
To construct the $\htwo$ model, we start with a finite undirected weighted complete graph
$G$ without self-loops 
with vertex set $\Lambda\cup\{\rho\}$, edge set $E$, and edge weights
$W_e\ge 0$, $e\in E$. The special vertex $\rho\notin\Lambda$ is called the pinning vertex;
the corresponding edge weights $h_i:=W_{i\rho}=W_{\{i,\rho\}}$, $i\in\Lambda$, are
called pinning strengths. We assume that the subgraph 
\begin{align}
  \label{eq:def-G-plus}
  G_+=(\Lambda\cup\{\rho\},E_+) \quad\text{with}\quad
  E_+:=\{ e\in E: W_e>0\}
\end{align}
is connected; it is obtained by dropping all edges
with zero weights.

We assign to every vertex $i\in\Lambda$ a superspin vector $v_i\in\htwo$ which is
described in horospherical coordinates by $(u_i,s_i,\overline \psi_i,\psi_i)$.
In addition, we set $(u_\rho,s_\rho,\overline \psi_\rho,\psi_\rho)=0$, which corresponds
to the reference point $(x_\rho,y_\rho,z_\rho,\xi_\rho,\eta_\rho)=(0,0,1,0,0)=:o$.
The action of the $\htwo$ model is defined by
\begin{align}
  \label{eq:def-A}
  A:=\sum_{e\in E}W_e(S_e-1)\quad\text{with }
  S_{ij}=S_{\{i,j\}}:=-\sk{v_i,v_j},
  \end{align}
which reads in horospherical coordinates
\begin{align}
  \label{eq:def-Sij-Bij}
  S_{ij}=B_{ij}+(\overline\psi_i-\overline\psi_j)(\psi_i-\psi_j)e^{u_i+u_j}\text{ with }
B_{ij}:=\cosh(u_i-u_j)+\frac12(s_i-s_j)^2e^{u_i+u_j};
\end{align}
cf.\  \cite[eq.\ (2.9-2.10)]{disertori-spencer-zirnbauer2010}.
To stress the difference between internal edges $e=\{i,j\}\subseteq\Lambda$
and pinning edges $\{i,\rho\}$, we write $W=(W_e)_{e\in E: e\subseteq\Lambda}$
and $h=(h_i)_{i\in\Lambda}$. As in 
\cite[eq.\ (2.12-2.13)]{disertori-spencer-zirnbauer2010}, the superintegration form of
the model is given by
\begin{align}
  \label{eq:def-superintegration}
  \sk{f} = \sk{f}^{\Lambda }_{W,h}:=\int_{\R^\Lambda\times\R^\Lambda}
  \prod_{i\in \Lambda} \left(
  \frac{e^{-u_i}}{2\pi}\, du_i\, ds_i\,\partial_{\overline\psi_i}\partial_{\psi_i}\right)
  f e^{-A}
\end{align}
for any function $f\left((u_i,s_i,\overline\psi_i,\psi_i)_{i\in\Lambda}\right)$
such that the integral exists. 
For any integrable function $f=f(u,s)$ which does not depend on the
Grassmann variables, one has
\begin{align}
  \label{eq:def-expectation}
  \E[f]= \E^{\Lambda }_{W,h}[f]= \E^{\Lambda\cup\{\rho\}}_{W,h}[f]
  :=&\sk{f}=\int_{\R^\Lambda\times\R^\Lambda}f(u,s)\, \mu^{\Lambda }_{W,h}(du\, ds), 
\end{align}
where $\mu^{\Lambda }_{W,h}$ is the probability measure on $\R^\Lambda\times\R^\Lambda$ given by
\begin{align}
  \mu(du\, ds)=  \mu^{\Lambda }_{W,h}(du\, ds):= &
                                                   e^{-\sum_{e\in E_+}W_e(B_e-1)}\det D
  \prod_{i\in \Lambda}\frac{e^{-u_i}}{2\pi}\, du_i\, ds_i,
   \label{eq:def-mu}
\end{align}
and $D=(D_{ij})_{i,j\in \Lambda}$ is the weighted Laplacian matrix on $G_+$
with entries 
\begin{align}
  \label{eq:representation-D-with-pinning-old}
  D_{ij}=\left\{
  \begin{array}{ll}
    -W_{ij}e^{u_i+u_j} & \text{if }i\neq j, \\
    \sum_{k\in(\Lambda\cup\{\rho\})\setminus\{i\}}W_{ik}e^{u_i+u_k}
    =\sum_{k\in \Lambda\setminus\{i\}}W_{ik}e^{u_i+u_k}+h_ie^{u_i}& \text{if }i=j,
  \end{array}
  \right.
\end{align}
see \cite[(2.9)]{disertori-merkl-rolles2020}; the $u$-marginal of $\mu$,
which is obtained by evaluating the Gaussian integral in the $s$-variables, is 
also given in \cite[(1.2) and (1.5)]{disertori-spencer-zirnbauer2010}.
Since $G_+$ is connected, the matrix $D$ is positive definite and hence invertible.
By supersymmetry, $\mu^{\Lambda }_{W,h}$ is a probability measure, see 
\cite[Proposition 2 in Appendix C]{disertori-spencer-zirnbauer2010}.  

\paragraph{Aim of this paper.}
We consider a version of $\htwo$ with hierarchical weights. Our main result, stated in Theorem \ref{thm:hierarchical}, is a bound on the expectation of powers of
\begin{align}
  B_{i\rho}=\cosh u_i+\frac12s_i^2e^{u_i},
\end{align}
i.e., the bosonic part of the hyperbolic cosine of the hyperbolic-distance observable between the superspin at $i$ and the reference 
superspin at $\rho.$  More precisely, we  prove that sufficiently strong hierarchical interactions force
moments of $B_{i\rho}$ to be arbitrarily close to $1,$ uniformly in the size of the graph and in the choice of the 
vertex $i.$ Unlike the earlier estimates used in \cite{disertori-merkl-rolles2023}, this localizes $u_i$ and $s_i$ jointly.
The main technical tools of the proof are: the reduction of the problem to an effective one-dimensional model with non-homogeneous weights, and scale-dependent bounds for determinants like $\det D$ in \eqref{eq:def-mu}.

\subsection{Hierarchical model}
\label{sec:hierarchical-model}
Let $N\in\N$. We consider the $\htwo$ model on the complete graph with vertex
set $\Lambda_N=\{0,1\}^N$ and hierarchical interactions. We interpret
$\Lambda_N$ as the set of
leaves of the binary tree $\T=\cup_{n=0}^N\{0,1\}^n$. 
For $i=(i_0,\ldots,i_{N-1}),j=(j_0,\ldots,j_{N-1})\in\{0,1\}^N$, let
\begin{align}
  d_H(i,j)=\min(\{l\in\{0,\ldots,N-1\}:\, i_k=j_k\text{ for all }k\ge l\}\cup\{N\})
\end{align}
be the hierarchical distance between $i$ and $j$, which equals 
the distance to the least common ancestor in $\T$. Let $w^H:[1,\infty)\to (0,\infty)$.
We choose hierarchical weights
\begin{align}
  \label{eq:hierarchical-weights}
  W_{ij}^H=w^H(d_H(i,j)), \quad i,j\in\Lambda_N, i\neq j, 
\end{align}
and uniform pinning $W_{i\rho}=h^H>0$
for all $i\in\Lambda_N$.

For this hierarchical $\htwo$ model, we prove the following localisation result,
which concerns both families of variables, $u$ and $s$ jointly.
\begin{theorem}[Bound for the hierarchical model]
  \label{thm:hierarchical}
Consider the above hierarchical model with the weight function 
  $w^H$ and the pinning $h^H$ fulfilling 
  \begin{align}
    \label{eq:ass-weights-pinning}
  w^H(r)\ge 8\overline W 2^{-2r}r^\alpha \text{ for } r\in\N, \quad
  h^H\ge 2\overline W2^{-N}(N+1)^\alpha 
  \end{align}
  with $\alpha>3$ and $\overline W>0$.  
  Assume that $\kappa\in(0,1]$ 
  and $m\ge 0$ fulfill the inequalities 
   \begin{align}
     \label{eq:assumptions-kappa-overline-W-m}
     \kappa\le\kappa(\alpha):=\frac{1}{16\ceins\log 2},\qquad 
     \overline W\ge \frac{1}{2\ceins\kappa}\log_2\frac{36}{\kappa},\qquad
     m\le\ceins\kappa\overline W
   \end{align}
   with the constant $\ceins=\ceins(\alpha):=\frac23 e^{-\frac{2\sqrt{2}}{\alpha-3}}$. 
  Then, for all $N\in\N$ and $i\in\Lambda_N$, one has  
\begin{align}
  \label{eq:hierarchical-est-cosh-ui}
  \E_{W^H,h^H}^{\Lambda_N}[\body(z_i)^m]
  =\E_{W^H,h^H}^{\Lambda_N}[B_{i\rho}^m]\le 1+\kappa, 
\end{align}
where $\body(z_i)=\cosh u_i+\frac12s_i^2e^{u_i}=B_{i\rho}$ is the unique real number such that $z_i-\body(z_i)$ is nilpotent.
Here are two examples of possible choices of the parameters.
  \begin{itemize}
  \item For $s\in (0,1)$, we may choose $\kappa=\overline W^{-s}$ with $\overline W$ large
  enough, depending on $\alpha$ and $s$, which allows us to take $m\in[0,\ceins\overline W^{1-s}]$.
\item  Another choice is $\kappa=\czwei\frac{\log\overline W}{\overline W}$
  with $\overline W$ large enough, depending on $\alpha$, with a constant $\czwei>1/(2\ceins\log 2)$, which allows us to take 
  $m\in[0,\ceins\czwei\log\overline W]$. 
  \end{itemize}
\end{theorem} 
  
Note that increasing $\overline W$ increases the exponent  $m$ that we can control. 
The observable $\body(z_i)=-\body\sk{v_i,o}$ can be interpreted as the hyperbolic
cosine of the hyperbolic distance
between the superspin $v_i$ and the reference superspin $o=(0,0,1,0,0)\in\htwo$.

\begin{corollary}
  \label{cor:hierarchical-differences}
  Under the assumptions of Theorem~\ref{thm:hierarchical}, for $\kappa,\overline W,m$
  as in Theorem~\ref{thm:hierarchical}, and for all $i,j\in\Lambda_N$, one has
 \begin{align}
   \label{eq:hierarchical-est-cosh-ui-uj}
   \E_{W^H,h^H}^{\Lambda_N}[\body(-\sk{v_i,v_j})^{m/2}]
   =\E_{W^H,h^H}^{\Lambda_N}[B_{ij}^{m/2}]\le 2^{m/2}(1+\kappa).
  \end{align}
\end{corollary}

\paragraph{Remark.}
 One could get rid of the factor $2^{m/2}$ in \eqref{eq:hierarchical-est-cosh-ui-uj}
by connecting the vertices $i$ and $j$ by an antichain going twice over all levels
below the least common ancestor in the binary tree. This antichain is described in
Remark~4.2 and visualized in Figure 4.3 in
\cite{disertori-merkl-rolles2020}. As a result,
one obtains a one-dimensional effective model with nearest-neighbor weights that increase up to the middle
and decrease afterwards. The strategy of the proof of Theorem~\ref{thm:hierarchical}
can be adapted to this case.

\subsection{Inhomogeneous one-dimensional model}
\label{sec:onedimensionaleff}
Using \cite[Cor.~2.3]{disertori-merkl-rolles2020}, we show in \eqref{eq:hierantichain} below that the expectation of $B_{i\rho}^m$ in the hierarchical model can
be reduced to
a similar expectation in an effective inhomogeneous one-dimensional model
with vertex set $\{1,\ldots,N\}$ and long-range interactions.
The vertices in this effective model
 correspond to length scales in the original hierarchical model.
The following model describes a one-dimensional chain with inhomogeneous weights and purely
nearest neighbor interactions. In the application to our hierarchical model, additional long-range interactions will be included as well, but these can be easily removed via a monotonicity result given in Lemma~\ref{le:monotonicity}.

\paragraph{One-dimensional model with inhomogeneous interactions.}
For any $N\in\N$, we consider the graph
$([N],E_N)$ with vertex set $[N]:=\{1,2,\ldots,N\}$,
undirected edge set $E_N=\{\{i-1,i\}:2\le i\le N\}$, and weights
$W_{\{i-1,i\}}\ge\overline Wi^\alpha$, $2\le i\le N$, depending on positive parameters $\overline W$
and $\alpha$. Later these two parameters will coincide with the corresponding
parameters in the hierarchical model from Section \ref{sec:hierarchical-model}.
We consider two versions of pinning.
\begin{enumerate}
\item[P1.] Either the pinning point $\rho$ does not belong to $[N]$.
  In this case, we take any pinning vector 
$h=(h_j)_{j=1,\ldots,N}\in[0,\infty)^N\setminus\{0\}$. 
\item[P2.] Or the pinning point $\rho$ coincides with one of the vertices in 
  $[N]$, e.g., $\rho=N$. Here, we take the pinning vector
  $h=(h_j)_{j=1,\ldots,N,j\neq\rho}$ with $h_j=W_{j\rho}$. 
\end{enumerate}

The following theorem is the main technical
 result of this paper. It gives bounds on this family of inhomogeneous one-dimensional models,
 possibly with additional long-range interactions.
We consider two sets of edges $E'$ and $E''$.
Edges in $E'$ are nearest-neighbor, while edges in $E''$ are typically
not, see remarks below the theorem. 
For any edge $e$ in the complete graph over $[N]$, we write $e =\{x_e,y_e\}$ with $x_e<y_e$. 

\begin{theorem}[Bounds for the inhomogeneous one-dimensional model]
  \label{thm:one-dim}
Let $\kappa\in(0,1]$, $\alpha>3$ and $\gamma\ge 0$ with $\alpha-\gamma>3$.
There exist $\wnull=\wnull(\kappa,\alpha,\gamma)>0$ and
  $\cdrei=\cdrei(\kappa,\alpha,\gamma)>0$ such that for all
  $\overline W\ge\wnull$ the following hold for the above model.
Set $m_x:=\cdrei\overline W x^\gamma$ for all $x\in\N$. 
\begin{enumerate}
\item   For every $e=\{x,y\}$ with 
  $x,y\in[N]$ and $x<y$, and every $m\in[0,m_x]$, one has
  \begin{align}
    \label{eq:bound-cosh-ux-uy-main}
    \E_{W,h}^{[N]}[(\cosh(u_x-u_y)+\tfrac12(s_x-s_y)^2e^{u_x+u_y})^m]=\E_{W,h}^{[N]}[B_e^m]\le 1+\kappa.
  \end{align}
  \item More generally, let  $E',E''$ be disjoint subsets of the edge set
  $\{\{x,y\}\subseteq[N]:x\neq y\}$ of the complete graph over $[N]$.
  We assume that every $e\in E'$
  is a nearest neighbor edge, $y_e-x_e=1$, i.e., $E'\subseteq E_N$.
  Moreover, we assume that for any two
different $e,\tilde e\in E'\cup E''$ the open intervals
$(x_e,y_e)$ and $(x_{\tilde e},y_{\tilde e})$ are disjoint.

Then, for all exponents $m_e$ satisfying both inequalities $0\le m_e<W_e$ and
$(1-\frac{m_e}{W_e})^{-1}\le 1+\kappa$ for $e\in E'$, and 
$0\le m_e\le m_{x_e}$ for $e\in E''$, 
we have   
\begin{align}
  \label{eq:main-thm-bound1}
  \E_{W,h}^{[N]}\left[\prod_{e\in E'}B_e^{m_e}\cdot\prod_{e\in E''}B_e^{m_e}
  \right]\le\prod_{e\in E'}\frac{1}{1-\frac{m_e}{W_e}}\cdot (1+\kappa)^{|E''|}
 \le (1+\kappa)^{|E'|+|E''|},
\end{align}
uniformly in $N$.
\end{enumerate}
In the case P1 of pinning, all estimates are also uniform in the pinning $h$.
The same results hold if we add arbitrarily many additional non-nearest
neighbor edges~$e$ with arbitrary non-negative weights $W_e\ge 0$
to the graph $([N],E_N)$. 
\end{theorem}
A quantitative version of this theorem is given in Lemma \ref{le:thm-quantitative},
below. 

Note that the exponent $m_e$ with $e=\{x_e,y_e\}$ may increase as $x_e$
increases, due to the inhomogeneous nature of the nearest-neighbor weights
in the model. 

Let us explain the different roles of $E'$ and $E''$:
  Edges $e\in E'$ can be treated in a simple way using supersymmetric Ward identities
  (see equation \eqref{eq:cite-from-dmr23}) resulting in the bound
  $(1-m_e/W_e)^{-1}$;
  however, these edges need to have a positive weight $W_e>0$
  with $m_e/W_e<1$.
  For edges with small positive $W_e$ this provides a bound only for 
  exponents $m_e$ so small that it is not useful anymore. For edges with zero weight
  the strategy does not work at all. Edges of this type form the set $E''$.
  The corresponding estimate is the most technical part of this work and 
  requires the weights of nearest-neighbor edges to be large enough.
  This technical construction may also be applied to nearest-neighbor edges, but
  the corresponding bound is weaker, i.e., $1+\kappa$ instead of $(1-m_e/W_e)^{-1}$. 
  
\subsection{Discussion}
Let us compare Theorem \ref{thm:hierarchical} with \cite[Theorem 3.2]{disertori-merkl-rolles2023}.
The latter treats the same model, considering only the $u$-variables, but not
the $s$-variables, under the weaker assumption $\alpha>1$
rather than $\alpha>3$. It shows that 
$\E^{\Lambda_N}_{W,h}[e^{\pm m u_i}]\le C=\exp(2^{-3}\overline W^{-1}(2m+1)^2\sum_{l=2}^\infty l^{-\alpha})$
for all $\overline W>0$ and $m\ge 1$.
In contrast, in the present paper, our bound for the expectation of
the larger quantity $(\cosh u_i+\frac12s_i^2e^{u_i})^m=\body(z_i)^m$
is asymptotically arbitrarily close to 1 as $\overline W\to\infty$,
which means localization of the $u_i$ and $s_i$ jointly near 0. 
Note that $o=(0,0,1,0,0)\in\htwo$ is the only superspin, up to nilpotent modifications,
with $u=0$ and $s=0$. In the language of the $\htwo$ model, our result means that
the superspins localize close to $o$. Because $o$ is not invariant under Lorentz
transforms $x'=\cosh(t)x+\sinh(t)z$, $y'=y$, $z'=\sinh(t)x+\cosh(t)z$, 
$\xi'=\xi$, $\eta'=\eta$ with a parameter $t\in\R\setminus\{0\}$, 
this localization also implies a 
breaking of Lorentz symmetry. Note that our bound for the pinning
$h^H$ given in \eqref{eq:ass-weights-pinning} 
is roughly on the scale $1/\text{volume}=2^{-N}$, up to a lower-order
correction, logarithmic in the volume. In particular, this bound vanishes in the limit as $N\to\infty$.
In this sense, the localization of the superspins close to $o$ becomes
a spontaneous breaking of Lorentz symmetry in the thermodynamic limit
$N\to\infty$ if we take the pinning strength on the scale of this bound. 

The proof relies heavily on the supersymmetry of the model. 
Although the condition $\alpha>3$ appears only in a technical way,
cf.~\eqref{eq:bound-sum-delta},
we do not know whether our localization result jointly in $u_i$ and $s_i$
holds also for $1<\alpha\le 3$.

The strategy and some of the techniques used here are similar to the ones used in
\cite{disertori-spencer-zirnbauer2010}. In particular,
the Ward identity \cite[(8.4)]{disertori-spencer-zirnbauer2010} plays the same role
as \eqref{eq:ward-id-with-protection}.
The determinant lower bounds in Section \ref{sec:bounds-on-dets} extend
  analogous bounds for $\Z^d$ appearing in the proofs of Lemma 6 and 7 in
  \cite{disertori-spencer-zirnbauer2010}, in particular equation (8.10),
  to general graphs.
To the best of our knowledge, the general
technique of separating vertices to bound determinants as treated in Lemma \ref{le:collapse}
is new. 
As in \cite{disertori-spencer-zirnbauer2010}, the notion of protected and unprotected pairs
of vertices coupled with a partition of unity is a key ingredient. 
However, in that paper the edge weights
are homogeneous on $\Z^d$, and an inductive construction on length scales is necessary.
This requires, in particular, introducing bad and good points at different length scales. 
In the present paper, due to the hierarchical structure, we can reduce the problem
to a one-dimensional model with
inhomogeneous edge weights whose vertices represent length scales. For this, we apply a
representation as an effective model introduced in \cite{disertori-merkl-rolles2020}.
As a result, we can bound the fluctuations in a single step, completely avoiding 
both an inductive argument and a distinction between good and bad points.
The same techniques from \cite{disertori-spencer-zirnbauer2010} have also been used
in \cite{disertori-sabot-tarres2014}; there it was not possible to avoid the
inductive argument.

\paragraph{How this paper is organized.}
The bounds for the one-dimensional model stated in Theorem \ref{thm:one-dim} are
proven in Section \ref{sec:bounds-onedimensionaleff}, using supersymmetry of the model.
A technical key point is a lower bound for a determinant arising from a
Grassmann Gaussian integral. This bound is provided in Section \ref{sec:bounds-on-dets},
even in a more general context than needed for this paper, since it is of independent
interest. Theorem \ref{thm:hierarchical} and Corollary \ref{cor:hierarchical-differences}
are then proven in Section \ref{sec:bounds-hier-lr}. 
  
The constants $\wnull$ and $\ceins,\ldots,\cvier$ keep their meaning throughout the whole
paper.

\section{Bounds in the one-dimensional model}
\label{sec:bounds-onedimensionaleff}

\subsection{Notation}
In this section, we consider again the setup of Section \ref{sec:susy} with
the complete graph $G=(\Lambda\cup\{\rho\},E)$. 
Recall also the definition \eqref{eq:def-G-plus} of the edge set $E_+$.
For $e\in E$, let $e_+$ and $e_-$ denote the two endpoints.
For bookkeeping reasons only, we give $e$ the direction $e_+\to e_-$.
The signed incidence matrix of the graph $G$ with the wiring point $\rho$ removed
but including the pinning edges is given by
\begin{align}
  \label{eq:def-F}
F\in\{-1,0,1\}^{\Lambda\times E}, \quad F_{k,e}=1_{\{ e_+=k\}}-1_{\{ e_-=k\}} 
  \quad\text{for } k\in \Lambda,\, e\in E. 
\end{align}
Using this notation, we can represent the weighted Laplacian $D=D(u)$ 
from
\eqref{eq:representation-D-with-pinning-old} as 
\begin{align}
  \label{eq:def-D}
  D=F\W F^t\in\R^{\Lambda\times \Lambda}
\end{align}
with the diagonal matrix 
\begin{align}
  \W=\W(u):=\diag(\W_e,e\in E)\in\R^{ E\times E}
  \quad \text{with}\quad\W_e:=W_ee^{u_{e_+}+u_{e_-}}.
\end{align}
For $m_e\ge 0$, $e\in E$, we define the diagonal matrices
\begin{align}
  \label{eq:def-M-G}
  &M:=\diag(m_e,e\in E), \\
  \label{eq:def-Q}
  &Q=Q^G(u,s):=\diag(Q_e,e\in E) \quad\text{with} \quad Q_e:=\frac{e^{u_{e_+}+u_{e_-}}}{B_e(u,s)}
\end{align}
with $B_e$ given in \eqref{eq:def-Sij-Bij}. Furthermore, we introduce the matrix 
\begin{align}
   &\G=\G^G_W(u,s):=\sqrt{Q}F^tD^{-1}F\sqrt{Q}\in\R^{ E\times E}.
    \label{eq:def-G}
\end{align}
Note that the inverse $D^{-1}$ exists because the edge set $E_+$ is connected. 

\paragraph{Interpretation.}
Consider an electrical network $G$ consisting of nodes in $\Lambda\cup\{\rho\}$
and conductance $\W_e$ for every edge $e\in E$. The node $\rho$ is
grounded. If an electrical potential vector $U=(U_i)_{i\in\Lambda}$
is applied to the nodes, then $DU$ is the vector of external currents
flowing into the nodes $i\neq\rho$. As a consequence, the matrix $D^{-1}$ can
be interpreted as the resistance matrix of this electrical network. 
For $x\in\Lambda$, let $1_x:=(\delta_{xj})_{j\in\Lambda}$ denote the $x$-th canonical unit
vector in $\R^\Lambda$ and set $1_\rho=0$. 
For $e=\{x,y\}\in E$, one has
\begin{align}
  (1_x-1_y)^t D^{-1}(1_x-1_y)=\cR^G(\W,x\leftrightarrow y),
\end{align}
where $\cR^G(c,x\leftrightarrow y)$ denotes the effective resistance
between $x$ and $y$ in the electrical network $G$ with conductances
$c_e$, $e\in E$. Consequently, 
\begin{align}
  \label{eq:G-ee-as-resistance-ex}
  \G_{ee}=&\sqrt{Q_e}(1_x-1_y)^t D^{-1}(1_x-1_y)\sqrt{Q_e}
                     =\cR^G(Q_e^{-1}\W,x\leftrightarrow y).
\end{align}

\subsection{An illustrative example}
\label{subsec:illustrative-ex}
Before explaining the proof of Theorem \ref{thm:one-dim} in its full complexity,
we present first an illustrative example in a sketchy way, concentrating on
the main points. We begin with some general remarks. The starting point
is \cite[Lemma 4.2]{disertori-merkl-rolles2023}, which states
\begin{align}
  \label{eq:ward-identity-with-det}
  1=\E\left[\prod_{e\in E'}B_e^{m_e}\det(\Id-M\G)\right].
\end{align}
Its proof relies on supersymmetry and an integration of the Grassmann variables.
A variant of this statement which includes cutoff functions is given in
\eqref{eq:integrate-grassmann} and \eqref{eq:ward-det-explicit} below; see also
\eqref{eq:rel-det-id-min}.

If $0< m_e<W_e$ for all $e\in E'$,
\begin{align}
  \det(\Id-M\G)\ge \prod_{e\in E'}\left(1-\frac{m_e}{W_e}\right)
  \label{eq:first-lower-bound-det}
\end{align}
holds,
cf.\ \cite[(4.29)]{disertori-merkl-rolles2023}. This leads to
\begin{align}
  \E\left[\prod_{e\in E'}B_e^{m_e}\right]\le\prod_{e\in E'}\frac{1}{1-\frac{m_e}{W_e}},
  \label{eq:cite-from-dmr23}
\end{align}
cf.\ \cite[(4.18)]{disertori-merkl-rolles2023}, 
which coincides with
claim \eqref{eq:main-thm-bound1} in the special case $E''=\emptyset$.
However, inequality \eqref{eq:first-lower-bound-det}
may fail to hold if $0\le W_e<m_e$ for some edges;
the determinant $\det(\Id-M\G)$ might even take negative values.

For our illustrative example, we study $\E[B_e^{m_e}]$. Our goal is to show that
it is not much bigger than one. More specifically, we
consider the case of a single
edge $e=\{x,y\}$ with $m_e>0$ and, say, $W_e$ being zero or too small
to be useful for an estimate as described above. 
In this simplified illustrative example, let us assume that
$x$ and $y$ are connected by a path consisting of two edges $\epsilon_1=\{x,z\}$
and $\epsilon_2=\{z,y\}$ with positive weights $W_{\epsilon_1}$ and $W_{\epsilon_2}$
fulfilling $m_e<W_{\epsilon_1}$ and $m_e<W_{\epsilon_2}$.
We take two positive parameters $\delta_1$ and $\delta_2$, to be tuned later. 
We consider the indicator functions $\chi_1=1_{\{B_{\epsilon_1}< 1+\delta_1\}}$
and  $\chi_2=1_{\{B_{\epsilon_2}< 1+\delta_2\}}$
of the events that $B_{\epsilon_1}$ and $B_{\epsilon_2}$, respectively, are close to $1$.
The complementary indicator functions are denoted by
$\chi_1^c=1-\chi_1$ and $\chi_2^c=1-\chi_2$. We use the partition of
unity
\begin{align}
  1=\chi_1\chi_2+\chi_1^c\chi_2+\chi_2^c,
\end{align}
multiply it by $B_e^{m_e}$ and use the Chebyshev-type inequalities
\begin{align}
  \label{eq:chebyshev-type-inequ}
  \chi_j^c=1_{\{B_{\epsilon_j}\ge 1+\delta_j\}}\le
  (1+\delta_j)^{-p_j}B_{\epsilon_j}^{p_j}, \quad (j=1,2)
\end{align}
with positive parameters $p_1$ and $p_2$, also to be tuned later, to obtain the bound
\begin{align}
  \label{eq:split-B_e^m}
  B_e^{m_e}&= B_e^{m_e}\chi_1\chi_2+B_e^{m_e}\chi_1^c\chi_2+B_e^{m_e}\chi_2^c
             \notag
  \\&\le B_e^{m_e}\chi_1\chi_2+(1+\delta_1)^{-p_1}B_e^{m_e}B_{\epsilon_1}^{p_1}\chi_2
  +(1+\delta_2)^{-p_2}B_e^{m_e}B_{\epsilon_2}^{p_2}.
\end{align}
By \cite[Lemma~2 in Section 5.2]{disertori-spencer-zirnbauer2010},
  one has
  \begin{align}
    B_e\le 2B_{\epsilon_1}B_{\epsilon_2}. 
    \label{eq:triangle-strat}
  \end{align}
  To interpret this formula, one may view 
  \begin{align}
    \varphi:=\arcosh(B_e),\quad
    \varphi_1:=\arcosh(B_{\epsilon_1}),
    \quad\text{ and }\quad
    \varphi_2:=\arcosh(B_{\epsilon_2}), 
  \end{align}
  as the hyperbolic distance between the spin variables indexed by
  $x$ and $y$, by $x$ and $z$, and by $z$ and $y$, respectively.
  These hyperbolic distances fulfill the triangle inequality $\varphi\le
  \varphi_1+\varphi_2$. From this point of view,
  inequality \eqref{eq:triangle-strat} gets the interpretation
  \begin{align}
    B_e=\cosh(\varphi)\le \cosh(\varphi_1+\varphi_2)
    \le 2\cosh(\varphi_1)\cosh(\varphi_2)=2B_{\epsilon_1}B_{\epsilon_2}.    
  \end{align}
  Substituting this in \eqref{eq:split-B_e^m}, we end up with
\begin{align}
  \label{eq:split-B_e^m2}
  B_e^{m_e}
  &\le B_e^{m_e}\chi_1\chi_2+
    2^{m_e}(1+\delta_1)^{-p_1}B_{\epsilon_1}^{m_e+p_1}B_{\epsilon_2}^{m_e}\chi_2
    +2^{m_e}(1+\delta_2)^{-p_2}B_{\epsilon_1}^{m_e}B_{\epsilon_2}^{m_e+p_2}.
\end{align}
Our goal is to show that the expectations of these three summands 
have upper bounds arbitrarily close to $1$, $0$, and $0$, respectively,
if the weights $W_{xz}$ and $W_{yz}$ are large enough
and the parameters $\delta_j$ and $p_j$ for $j=1,2$ are tuned appropriately.

\paragraph{Bounding the expectations of the second and third summand in \eqref{eq:split-B_e^m2}.}
For our simplified example, let us assume that $m_e+p_1<W_{\epsilon_1}$
and $m_e+p_2<W_{\epsilon_2}$ hold, and hence, we can apply bound \eqref{eq:cite-from-dmr23}.
We estimate the expectation of the second summand by
\begin{align}
  \frac{2^{m_e}}{(1+\delta_1)^{p_1}}\E\left[B_{\epsilon_1}^{m_e+p_1}B_{\epsilon_2}^{m_e}\chi_2\right]
  \le \frac{2^{m_e}}{(1+\delta_1)^{p_1}}\E\left[B_{\epsilon_1}^{m_e+p_1}B_{\epsilon_2}^{m_e}\right]
  \le \frac{2^{m_e}(1+\delta_1)^{-p_1}}{\left(1-\frac{m_e+p_1}{W_{\epsilon_1}}\right)\left(1-\frac{m_e}{W_{\epsilon_2}}\right)}.
  \label{eq:bound-exp-e1-e2}
\end{align}
The parameters $\delta_1$ and $p_1$ need to be chosen so that
the quantity $(1+\delta_1)^{-p_1}$ is small enough to make
the right-hand side of \eqref{eq:bound-exp-e1-e2} small 
whenever $W_{\epsilon_1}$ and $W_{\epsilon_2}$ are sufficiently large.
A similar argument holds for the third summand.

\paragraph{Bounding the expectations of the first summand in \eqref{eq:split-B_e^m2}.}
Let us sketch the argument. 
Note that $\chi_1\chi_2=1_U$ where $U$ denotes the event that
$B_{\epsilon_j}<1+\delta_j$ for $j=1,2$. The identity \eqref{eq:ward-identity-with-det}
has a generalisation phrased in Lemma \ref{le:upper-bound-prod-B-chi-det}, 
which is also proved using supersymmetry. It gives
\begin{align}
  \label{eq:1-geq-exp-B}
  1\ge\E\left[B_e^{m_e}1_U\det(\Id-M\G)\right].
\end{align}
Note that this is not an immediate consequence of \eqref{eq:ward-identity-with-det}
because $\det(\Id-M\G)$ might take negative values on $U^c$. 
In our example, $\det(\Id-M\G)=1-m_e\G_{ee}$. The matrix element $\G_{ee}$
has the interpretation of an effective resistance between $x$ and $y$ in
the electrical network consisting of nodes $x,y,z,\rho$ and edges $e,\epsilon_1,\epsilon_2$
together with three pinning edges connecting $x,y,z$ with $\rho$,
cf.~\eqref{eq:G-ee-as-resistance-ex}. 
Any two nodes $i$ and $j$ are connected by a random Ohmian resistor with
conductance
\begin{align}
  \label{eq:def-c-example}
  c_{ij}=\frac{\W_{ij}}{Q_e}=W_{ij}e^{u_i+u_j-u_x-u_y}B_e,
\end{align}
cf.~\eqref{eq:def-c-e}. On the event $U$, the sum $|u_x-u_z|+|u_z-u_y|$ is bounded
by a constant, depending on $\delta=(\delta_1,\delta_2)$; let us call
it $\log C_\delta$; cf.\ \eqref{eq:diff-u-s-le-C-delta}. Since $B_e\ge 1$,
this yields $c_{ij}\ge C_{\delta}^{-1}W_{ij}$ for $i,j\in\{x,y,z\}$. By Rayleigh's 
mono\-toni\-city principle, the effective resistance can only increase
if the conductances are decreased and if we remove the three pinning edges
between $x,y,z$ and $\rho$. Thus, we obtain
\begin{align}
  \label{eq:bound-Gee}
  \G_{ee}=\cR(c,x\leftrightarrow y)\le\cR(C_{\delta}^{-1}W,x\leftrightarrow y)
  =C_{\delta}\left(W_{xz}^{-1}+W_{yz}^{-1}\right);
\end{align}
here $\cR(c,x\leftrightarrow y)$ and $\cR(C_{\delta}^{-1}W,x\leftrightarrow y)$
denote the effective resistance between $x$ and $y$, where in the first case
the underlying network is the complete graph over $\{x,y,z,\rho\}$ with
conductances $c$ given in \eqref{eq:def-c-example} and in the second case the line graph
with vertex set $\{x,y,z\}$, edges $\{x,z\}$ and $\{y,z\}$ and conductances
$C_\delta^{-1}W_{ij}$. The parameters $\delta_1$ and $\delta_2$ are
tuned in such a way that
\begin{align}
  \label{eq:bound-det-id-minus-MG-example}
  \det(\Id-M\G)=1-m_e\G_{ee}\ge 1-m_eC_{\delta}\left(W_{xz}^{-1}+W_{yz}^{-1}\right)=:1-\ell>0.
\end{align}
Inserting this into \eqref{eq:1-geq-exp-B}, we obtain 
\begin{align}
  \E\left[B_e^{m_e}\chi_1\chi_2\right]=\E\left[B_e^{m_e}1_U\right]\le (1-\ell)^{-1}.
\end{align}
This bound becomes arbitrarily close to $1$ if the weights 
$W_{xz}$ and $W_{yz}$ are large enough.

\paragraph{Limitations of the illustrative example and strategy of the proof of Theorem~\ref{thm:one-dim}.}
The illustrative example catches some of the features of the full 
proof of Theorem \ref{thm:one-dim} below, but it fails to reflect others.
First, when a product $\prod_e B_e^{m_e}$ rather than only one factor $B_e^{m_e}$ is treated,
the partition of unity generalizing \eqref{eq:split-B_e^m} becomes more complex;
cf.\ formulas~\eqref{eq:j-th-summand-new},
\eqref{eq:j-th-summand-new-X}, and  \eqref{eq:prod-B-sumprod-X} below.
Second, the treatment of the second and third summands in
\eqref{eq:split-B_e^m2} is an oversimplification compared with the 
analogous treatment in Lemma~\ref{le:exp-one-product} below; one needs to use
the cutoff functions and apply arguments similar to those used in
the treatment of the first summand in \eqref{eq:split-B_e^m2}. 
Third, we want to lower bound the expectation
$\E[\prod_{l=1}^n B_{e_l}^{m_{e_l}}1_U\det(\Id-M\G)]$ with an appropriate
cutoff event $U$ 
by a product of $n$ factors. This requires splitting the determinant
$\det(\Id-M\G)$, at least in  form of a lower bound,
into a product of smaller-sized determinants $\prod_{l=1}^n\det(\Id -M_l\G_l)$
associated to subgraphs. 
Let us heuristically explain the idea.
View the graph $G$ as being composed of several subgraphs $G_l$, $1\le l\le n$.
The subgraphs may overlap partially; in particular the pinning vertex $\rho$
should belong to all of them. Except the pinning vertex, we split every vertex 
that belongs to several subgraphs into several vertices, each belonging to one
of the subgraphs.
The copies of vertices generated from a single vertex are at first
connected by edges with infinite conductance. This procedure does not
yet change the determinant $\det(\Id -M\G)$. However, when the infinite
weights are replaced by zero weights, the determinant decreases, and it then
factorizes as a product of determinants; cf.\ \eqref{eq:det-lower-bound-in-proof}
below. These splitting and monotonicity arguments are treated in Subsection
\ref{sec:bounds-on-dets}.

\subsection{Moments with cutoff}
To prepare the proof of the theorem, we introduce cutoff functions and
explain how to deal with their derivatives that arise in Grassmann
integration.

For $e\in E$, let $\delta_e\in(0,\infty]$ and 
\begin{align}
  \chi_e:\R\to\R,\quad\chi_e=1_{(-\infty,1+\delta_e)}
  \label{eq:introduction-chi-e}
\end{align}
be the indicator function
of the interval $(-\infty,1+\delta_e)$.
Let $\chi_e^n:\R\to[0,1]$, $n\in\N$, be
smooth approximations of $\chi_e$ such that $\chi_e^n(b)=1$ for $b\le 1$,
$\chi_e^n(b)=0$ for $b\ge 1+\delta_e$, $b\mapsto\chi_e^n(b)$ is 
decreasing, and $\chi_e^n\uparrow\chi_e$ as $n\uparrow\infty$. The notation
$\chi_e^n$ should not be confused with a power of $\chi_e$. 
Note that for $\delta_e=\infty$,
one has $\chi_e\equiv 1$. In this case, we choose $\chi_e^n\equiv 1$ for
all $n$. 

As a replacement for the hypothesis $0\le m_e<W_e$ of bound
\eqref{eq:cite-from-dmr23}, we use the following
assumption, which needs to be checked in a model-dependent way
in concrete cases:
\begin{assumption}[Positive definiteness on a cutoff event]
  \label{ass:m-W}
  \mbox{}\\
  We assume that the choice of the parameters $\delta_e$, $e\in E$, depending
  on the weights $W_e$ and the exponents $m_e$, $e\in E$, is such that 
  $D(u)-FM Q(u,s)F^t>0$ holds on the set of configurations 
\begin{align}
\label{eq:def-U}
U:=\{(u,s)\in\R^\Lambda\times\R^\Lambda:B_e(u,s)< 1+\delta_e\text{ for all }e\in E\}.
\end{align}
\end{assumption}
Note that $1_U=\prod_{e\in E}\chi_e(B_e)$. 

The following lemma is a crucial ingredient in the proof of Theorem \ref{thm:one-dim}.
We do not have a universal
analogue of bound \eqref{eq:cite-from-dmr23} since
bounds on the determinant $\det(\Id-M\G)$ need to be derived for the individual
  models.

\begin{lemma}[Moments with cutoff and determinant]
  \label{le:upper-bound-prod-B-chi-det}
  Under Assumption \ref{ass:m-W}, one has
\begin{align}
  1\ge \E\left[\prod_{e\in E}B_e^{m_e}\cdot 1_U\cdot\det(\Id-M\G)\right].
  \label{eq:upper-bound-prod-B-chi-det-claim}
\end{align}
\end{lemma}
\begin{proof}
  The proof follows the lines of the proof of \cite[Lemma~4.2]{disertori-merkl-rolles2023}
  with some additional complications due to the presence of the cutoff
  functions. A similar statement appears in \cite[formula (8.8)]{disertori-spencer-zirnbauer2010}. The  Ward identity from  \cite[Lemma~4.1]{disertori-merkl-rolles2023}, which is a special case of
  the localization result
  \cite[Prop. 2 in the appendix]{disertori-spencer-zirnbauer2010},
  states
  \begin{align}
    \label{eq:ward-prod-f-S-e}
    \sk{\prod_{e\in E}f_e(S_e)}=\prod_{e\in E}f_e(1)
  \end{align}
  for any family $f_e:[1,\infty)\to \R$, $e\in E$, of continuously
  differentiable functions such that $f_e$ and $f_e'$ are bounded by
  polynomials.
  For any given $n\in\N$, we apply it 
  with $f_e(x)=x^{m_e}\chi_e^n(x)$,  
  which fulfills $f_e(1)=\chi_e^n(1)=1$. We obtain the following
  ``protected'' version of the Ward identity:
  \begin{align}
  \label{eq:ward-id-with-protection}
    \sk{\prod_{e\in E}(S_e^{m_e}\chi_e^n(S_e))}=1. 
  \end{align}
  Note that the Grassmann derivatives contained in the expectation involve
  derivatives of the cutoff functions. This is why we need to replace
  the indicator functions $\chi_e$ by smooth approximations $\chi_e^n$. 
  At the end we will take the limit $n\to\infty$. Consider an edge $e=\{i,j\}\in E$. 
  By \eqref{eq:def-Sij-Bij}, we decompose $S_e$ into its body $B_e$
  and a nilpotent part as follows.
  \begin{align}
   S_e=&B_e+(\overline\psi_i-\overline\psi_j)(\psi_i-\psi_j)e^{u_i+u_j}
         =B_e+e^{u_i+u_j}\overline\psi^tF_{\cdot e}F^t_{e\cdot}\psi
         \label{eq:Taylor-Se}
  \end{align}
  Here, the nilpotent part has square 0 already. As a consequence,
  only the zeroth- and first-order terms can contribute to the
  Taylor expansion of $S_e^{m_e}$ centered at its body $B_e^{m_e}$:
         \begin{align}
    S_e^{m_e}=&\left(B_e+e^{u_i+u_j}\overline\psi^tF_{\cdot e}F^t_{e\cdot}\psi\right)^{m_e}
                =B_e^{m_e}\exp\left(\overline\psi^tF_{\cdot e}m_eQ_eF^t_{e\cdot}\psi \right);
                \label{eq:taylor-exp-S-e-hoch-m-e}
  \end{align}
  cf.\ \cite[Equation (4.20)]{disertori-merkl-rolles2023} and recall the definition
  \eqref{eq:def-Q} of $Q_e$. 
  We expand $\chi_e^n(S_e)$ in the Grassmann variables:
\begin{align}
\chi_e^n(S_e)=& \chi_e^n(B_e)+(\chi_e^n)'(B_e)e^{u_i+u_j}\overline\psi^tF_{\cdot e}F^t_{e\cdot}\psi
\end{align}
Since $\chi_e^n\ge 0$, all $b$ with $\chi_e^n(b)=0$ are local
minima and hence satisfy $(\chi_e^n)'(b)=0$. Hence $\chi_e^n(B_e)=0$ implies
$\chi_e^n(S_e)=0$. Thus
$ S_e^{m_e}\chi_e^n(S_e)= S_e^{m_e}\chi_e^n(S_e)1_{\{\chi_e^n(B_e)>0\}}$. 
On the event $\{\chi_e^n(B_e)>0\}$, using
$(\overline\psi^tF_{\cdot e}F^t_{e\cdot}\psi)^2=0$, we obtain
\begin{align}
  \chi_e^n(S_e) =&\chi_e^n(B_e)\left(1+\overline\psi^tF_{\cdot e}R_eF^t_{e\cdot}\psi\right)
                   =\chi_e^n(B_e)\exp(\overline\psi^tF_{\cdot e}R_eF^t_{e\cdot}\psi)
\end{align}
with $R_e:=e^{u_i+u_j}(\chi_e^n)'(B_e)/\chi_e^n(B_e)$.
Combining this with \eqref{eq:taylor-exp-S-e-hoch-m-e} yields on the same event  
\begin{align}
  S_e^{m_e}\chi_e^n(S_e)
    = &B_e^{m_e}\chi_e^n(B_e)
        \exp\left(\overline\psi^tF_{\cdot e}[m_eQ_e+R_e]F^t_{e\cdot}\psi\right). 
\end{align}
The following calculations are performed on the event
$U_n:=\{\chi_e^n(B_e)>0\text{ for all }e\in E\}$. Note that
$U_n\uparrow U$ as $n\to\infty$. 
Let $R:=\diag(R_e,e\in E)$. 
One has 
\begin{align}
  \prod_{e\in E}(S_e^{m_e}\chi_e^n(S_e))
  =&\prod_{e\in E}(B_e^{m_e}\chi_e^n(B_e))\exp(-\overline\psi^t(-F(M Q+R)F^t)\psi).
\end{align}
Since $\chi_e^n$ is non-negative and decreasing, $(\chi_e^n)'/\chi_e^n\le 0$.
Consequently, 
\begin{align}
  -a^tFRF^ta=-\sum_{e=\{i,j\}\in E}e^{u_i+u_j}\frac{(\chi_e^n)'(B_e)}{\chi_e^n(B_e)}
  (a_i-a_j)^2\ge 0
\end{align}
for all $a\in\R^\Lambda$ and hence the matrix $-FRF^t$ is positive semidefinite. 

The Grassmann part in the Boltzmann factor $e^{-A}$ is given by $\exp(-\overline\psi^tD\psi)$. 
Hence, using that the product $\prod_{e\in E}\chi_e^n(B_e)$ vanishes on $U_n^c$ and
the value of the indicator function $1_{U_n}$ depends only
on the real variables $u$ and $s$, but not on the Grassmann variables
$\overline\psi$ and $\psi$, we can rewrite \eqref{eq:ward-id-with-protection} as follows, using the convention $0\cdot\text{undefined}=0$,   
\begin{align}
  1&=\sk{\prod_{e\in E}(B_e^{m_e}\chi_e^n(B_e))e^{-\overline\psi^t(-F(M Q+R)F^t)\psi}}
      \nonumber\\
  &=\int\limits_{\R^\Lambda\times\R^\Lambda}
     \prod_{e\in E}(B_e^{m_e}\chi_e^n(B_e))
     \prod_{i\in \Lambda}\partial_{\overline\psi_i}\partial_{\psi_i}
     e^{-\overline\psi^t(D-F(M Q+R)F^t)\psi}
  \, \frac{\mu(du\, ds)}{\det D}\nonumber\\
   &=\E\left[\prod_{e\in E}(B_e^{m_e}\chi_e^n(B_e))\frac{\det(D-F(M Q+R)F^t)}{\det D}\right].
     \label{eq:integrate-grassmann}
\end{align}
Since $U_n\subseteq U$, the matrix $D-FM QF^t$ is positive definite on $U_n$ by Assumption \ref{ass:m-W}.
Together with $-FRF^t\ge 0$ this yields
\begin{align}
  &\frac{\det(D-F(M Q+R)F^t)}{\det D}
    \ge\frac{\det(D-FM QF^t)}{\det D}
    =\det(\Id-D^{-1}FM QF^t)\nonumber\\
  =&\det(\Id-M \sqrt{Q}F^tD^{-1}F\sqrt{Q}) 
     =\det(\Id-M\G)>0,
     \label{eq:rel-det-id-min}
\end{align}
where we rewrote the product of diagonal matrices $M Q=\sqrt{Q}M \sqrt{Q}$
and used the fact $\det(\Id+AB)=\det(\Id+BA)$, which is valid for any
rectangular matrices $A\in\R^{m\times n}$, $B\in\R^{n\times m}$. We conclude
\begin{align}
  1\ge &\E\left[\prod_{e\in E}(B_e^{m_e}\chi_e^n(B_e))\det(\Id-M\G)\right]
     \label{eq:ward-det-explicit}
\end{align}
for all $n$.
Note that the expectation in \eqref{eq:ward-det-explicit} is taken 
in the usual classical measure theory; no Grassmann variables appear in it
anymore, as they have been integrated already in \eqref{eq:integrate-grassmann}.
The argument of the expectation in \eqref{eq:ward-det-explicit} is non-negative
and monotonically
increasing in $n$, because $0\le \chi_e^n\uparrow\chi_e$ as $n\to\infty$;
cf. the text following \eqref{eq:introduction-chi-e}.
Consequently, the monotone convergence theorem allows us to take the
limit as $n\to\infty$ in \eqref{eq:ward-det-explicit}.
From $\prod_{e\in E}\chi_e(B_e)=1_U$, we conclude that claim \eqref{eq:upper-bound-prod-B-chi-det-claim}
is valid.
\end{proof}

\subsection{Bounds on determinants}
\label{sec:bounds-on-dets}

The goal of this section is to derive lower bounds for the determinant $\det(\Id-M\G),$
under the cutoff event $U$ defined in \eqref{eq:def-U}. These bounds, which 
will be inserted in \eqref{eq:upper-bound-prod-B-chi-det-claim}, 
are given in terms of effective resistances of electrical networks, cf.\
\eqref{eq:bound-Gee} and \eqref{eq:bound-det-id-minus-MG-example} in the illustrative example.
The main result of this subsection is the bound \eqref{eq:estimate-det-summary} in Theorem~\ref{thm:summary-section}. 

Before going to the details, let us roughly explain the strategy.
First, we remove edges from the graph which we consider ``unimportant'',
i.e., we modify the corresponding weights to become 0. We show in
Lemma~\ref{le:monotonicity} that this only  decreases the corresponding
determinant. 
Second, the graph is split into several components that are
connected to one another
only through the pinning vertex $\rho$, as explained at the end of
Section \ref{subsec:illustrative-ex}. 

\smallskip
The following lemma is useful to verify Assumption \ref{ass:m-W}. 

\begin{lemma}[Characterisations of positive definiteness]
  \label{le:equivalence-of-hypotheses}\hspace{2cm}
  \begin{enumerate}
\item The following two statements are equivalent:
\begin{align}
  \sqrt{M}\G \sqrt{M}<\Id
  \qquad\Leftrightarrow\qquad
  D-FM QF^t>0.
           \label{eq:D-larger-FQF}
\end{align}
\item If $W_e\ge m_e$ for all $e\in E$ and the graph $(\Lambda\cup\{\rho\},\{e\in E:W_e>m_e\})$
  is connected, then the two equivalent conditions
  in \eqref{eq:D-larger-FQF} hold.
\end{enumerate}
\end{lemma}
\begin{proof}
For the first part, we use the following fact from linear algebra: for any possibly
    rectangular real-valued matrix $A$ the statements $AA^t<\Id$ and $A^tA<\Id$
    are equivalent. This follows from the fact that non-zero eigenvalues 
    of $AA^t$ and of $A^tA$ together with their multiplicities coincide. We apply this fact to $A=D^{-1/2}F\sqrt{M}\sqrt{Q}$.
    Note that $\sqrt{M}\sqrt{Q}=\sqrt{Q}\sqrt{M}$ because $M$ and $Q$ are diagonal matrices. 
    We calculate
    \begin{align}
    A^tA=\sqrt{M}\sqrt{Q}F^tD^{-1}F\sqrt{Q}\sqrt{M}=\sqrt{M}\G \sqrt{M};  
    \end{align}
    thus $A^tA<\Id$ is equivalent to the first statement in \eqref{eq:D-larger-FQF}. On the other
    hand, using the equality 
    \begin{align}
    AA^t=D^{-1/2}FMQF^tD^{-1/2}
    \end{align}
    and $D^t=D$, we obtain that $AA^t<\Id$ is equivalent to
    $D^{-1/2}(D-FMQF^t)D^{-1/2}>0$, which is equivalent to the second
    statement in \eqref{eq:D-larger-FQF}.    

    For the second part, assume that $W_e\ge m_e$ for all $e\in E$. Then it follows that 
    $W_e\ge m_e/B_e$ because $B_e\ge 1$ and
consequently
\begin{align}
  \W_e=W_ee^{u_{e_+}+u_{e_-}}\ge \frac{m_e}{B_e} e^{u_{e_+}+u_{e_-}}=m_eQ_e.
\end{align}
Moreover, the above inequality is strict if $W_e>m_e$. We assume now that the graph
$\tilde G:=(\Lambda\cup\{\rho\},\{e\in E:W_e>m_e\})$ is connected.
Let $x\in\R^\Lambda\setminus\{0\}$ and set $x_\rho=0$. Then, there is $i\in \Lambda$ with $x_i\neq 0$. Since
$i$ and $\rho$ are connected by $\tilde G$ there is an edge $e\in E$ with $W_e>m_e$ and
$x_{e_+}\neq x_{e_-}$, i.e.\ $(F^tx)_e\neq 0$. We conclude
\begin{align}
  0<(\W_e-m_eQ_e)(F^tx)_e^2\le x^tF(\W-MQ)F^tx.
\end{align}
Hence, $D=F\W F^t>FM QF^t$.
\end{proof}

The next result shows some monotonicity of $\det(\Id-M\G)$ with respect to $W_e$ and
$m_e$ and will be used to compare $\det(\Id-M\G)$ for $G$ and the
graph obtained from $G$ by removing some edges.

\begin{lemma}[Monotonicity of the determinant in the parameters]
  \label{le:monotonicity}
    Let $(m_e)_{e\in E}$, $(W_e)_{e\in E}\in[0,\infty)^E$ and
  $(m_e')_{e\in E},(W_e')_{e\in E}\in[0,\infty)^E$ be families of parameters
  such that  $m_e\le m_e'$ and
    $W_e\ge W_e'$ for all $e\in E$. Assume
    that $(m_e')_{e\in E}$ and $(W_e')_{e\in E}$ satisfy Assumption \ref{ass:m-W}.
    Then, $(m_e)_{e\in E}$ and $(W_e)_{e\in E}$ satisfy it as well.

    Moreover, let the matrices $M,\G$ and $M',\G'$ be defined according to
    \eqref{eq:def-M-G} and \eqref{eq:def-G}
    with the first and second family of parameters, respectively. Then, one has on the
    event $U$
    \begin{align}
      0<\det(\Id-M'\G')\le\det(\Id-M\G).
    \end{align}
\end{lemma}
\begin{proof}
Set $\W_e':=W_e'e^{u_{e_+}+u_{e_-}}$, $\W':=\diag(\W_e',e\in E)$, and $D':=F\W' F^t$.
  By assumption, we know $D'=F\W' F^t>FM'QF^t$ on the event $U$. The assumptions
  $m_e\le m_e'$ and $W_e\ge W_e'$, $e\in E$, imply $D=F\W F^t\ge F\W' F^t=D'>FM'QF^t\ge FMQF^t$,
  hence $D>FMQF^t$ and $D>FM'QF^t$ hold on the event $U$ as well
  and thus Assumption~\ref{ass:m-W} is satisfied.
  
  For the second claim, on the event $U$,
    we observe that the last two inequalities
    and our assumption imply $\sqrt M\G\sqrt M<\Id$, $\sqrt{M'}\G' \sqrt{M'}<\Id$, 
    and $\sqrt{M'}\G\sqrt{M'}<\Id$ by Lemma~\ref{le:equivalence-of-hypotheses}. 
    The last inequality and $M\le M'$ yield
    $0\le\sqrt{\G}M\sqrt{\G}\le\sqrt{\G}M'\sqrt{\G}<\Id$, cf.\
    first two lines in the
    proof of Lemma~\ref{le:equivalence-of-hypotheses}. 
    Since $D\ge D'$, we infer
  \begin{align}
    \G=\sqrt{Q}F^tD^{-1}F\sqrt{Q}\le \sqrt{Q}F^t(D')^{-1}F\sqrt{Q}=\G',
  \end{align}
  which implies $0\le\sqrt{M'}\G\sqrt{M'}\le\sqrt{M'}\G'\sqrt{M'}<\Id$.
  We conclude
\begin{align}
&  \det(\Id-M\G)=\det(\Id-\sqrt{\G}M\sqrt{\G})
                  \ge \det(\Id-\sqrt{\G}M'\sqrt{\G}) \nonumber\\
  &=\det(\Id-\sqrt{M'}\G\sqrt{M'})\ge\det(\Id-\sqrt{M'}\G' \sqrt{M'})=\det(\Id-M'\G')>0
\end{align}
because the determinant is monotonically increasing on the set of positive definite
matrices. 
\end{proof}

We now compare the determinants associated to two graphs $G$ and $G'$,
where $G'$ is obtained from $G$ by splitting some vertices and edges into several copies.
In other words, $G$ is obtained from $G'$ by identifying vertices according
to some map $k$. 

\begin{lemma}[Separating vertices]
  \label{le:collapse}
  Let $G=(\Lambda\cup\{\rho\},E)$ and $G'=(\Lambda'\cup\{\rho\},E')$ be
  two complete graphs, endowed with families of weights $W_e$, $W_{e'}'\ge0$,
  $e\in E$, $e'\in E'$, such that $G_+$ and $G_+'$ are connected, cf.\
  \eqref{eq:def-G-plus}. 
  We take parameters $m_e$, $m_{e'}'\ge 0$, $e\in E$, $e'\in E'$, and 
  we assume that there
  exists a surjective map $k:\Lambda'\cup\{\rho\}\to\Lambda\cup\{\rho\}$
  such that $k^{-1}(\{\rho\})=\{\rho\}$ and 
  \begin{align}
    W_e=\sum_{e'\in E'}W_{e'}'1_{\{k(e')=e\}},\quad
    m_e=\sum_{e'\in E'}m_{e'}'1_{\{k(e')=e\}},
    \label{eq:assumption-on-W-m-and-prime}
  \end{align}
  where for $e'=\{i,j\}$ we set $k(e'):=\{k(i),k(j)\}$. 
  Given $u_i$ and $s_i$ for $i\in \Lambda\cup\{\rho\}$, define
  $u_i':=u_{k(i)}$, $s_i':=s_{k(i)}$ for $i\in\Lambda'\cup\{\rho\}$. Let
  $\G(u,s),\G'(u',s'),M$, and
  $M'$ be defined as in \eqref{eq:def-G} and \eqref{eq:def-M-G}. 
  If $\sqrt{M'}\G'(u',s') \sqrt{M'}\le(1-\theta)\Id$ for a given
  $\theta\in(0,1)$, then $\sqrt{M}\G(u,s) \sqrt{M}\le(1-\theta)\Id$ and
  one has 
  \begin{align}
    \label{eq:claim-gluing}
    \det(\Id-M\G(u,s))\ge\det(\Id-M'\G'(u',s')).
  \end{align} 
\end{lemma}

The assumption $k^{-1}(\{\rho\})=\{\rho\}$
guarantees that the random variable $u_i'$ coincides with $u_\rho=0$ if
and only if $i=\rho$, and the same holds for $s_i'$ and $s_\rho=0$. 

\begin{proof}
  Throughout the proof, all primed matrices
  correspond to the weighted graph $(G',W')$ and are evaluated at $u',s'$.
  The claim \eqref{eq:claim-gluing} is equivalent to
\begin{align}
\label{eq:reduced-claim-gluing}
\det(\Id-\sqrt M\G\sqrt M)\ge\det(\Id-\sqrt{M'}\G'\sqrt{M'}).  
\end{align}
We define the matrix $K\in\{0,1\}^{\Lambda\times\Lambda'}$ by
$K_{jj'}=1$ iff $k(j')=j$. The proof uses the following two auxiliary identities 
\begin{align}
\label{eq:connection-prime-matrix1}
  KF'\W'(u')(F')^tK^t=& F\W(u)F^t, \\
  KF'M'Q'(u',s')(F')^tK^t=&FMQ(u,s)F^t,
\label{eq:connection-prime-matrix2}
\end{align}
which are shown in Lemma~\ref{le:aux} below. 
We introduce
\begin{align}
  C:=D^{-1/2}F\sqrt{QM}, \quad L:=D^{-1/2}KF'\sqrt{Q'M'}.
\end{align}
Note that $M$ and $Q$ are diagonal and hence commute. Using
\eqref{eq:connection-prime-matrix2} for the second equality, we calculate
\begin{align}
  \sqrt M\G\sqrt M=C^tC, \quad
  CC^t=LL^t.
\end{align}
It follows that $\sqrt M\G\sqrt M$, $LL^t$, and $L^tL$ all have the
same eigenvalues including multiplicities with the possible exception of 0, and 
\begin{align}
 \det(\Id-\sqrt M\G\sqrt M)=&\det(\Id-C^tC)=\det(\Id-CC^t)\nonumber\\
  =&\det(\Id-LL^t)=\det(\Id-L^tL).
     \label{eq:det-id-minus}
\end{align}
We will prove
\begin{align}
  \sqrt{M'}\G'\sqrt{M'}\ge L^tL.
  \label{eq:lower-bound-LL}
\end{align}
Since by assumption $\sqrt{M'}\G'\sqrt{M'}\le(1-\theta)\Id$, using
  the coincidence of eigenvalues $\neq 0$, this inequality implies
$L^tL\le(1-\theta)\Id$ and hence $\sqrt{M}\G\sqrt{M}\le(1-\theta)\Id$. The claim
\eqref{eq:claim-gluing} follows directly from \eqref{eq:det-id-minus}
and \eqref{eq:lower-bound-LL}.

To prove \eqref{eq:lower-bound-LL}, 
we calculate $L^tL=\sqrt{M'Q'}(KF')^tD^{-1}KF'\sqrt{Q'M'}$. Using
\eqref{eq:connection-prime-matrix1}, we find
$D=F\W F^t=KF'\W'(F')^tK^t=KD'K^t$ and thus
\begin{align}
  L^tL=\sqrt{M'Q'}(KF')^t(KD'K^t)^{-1}KF'\sqrt{Q'M'}.
\end{align}
The claim \eqref{eq:lower-bound-LL} is equivalent to
\begin{align}
  \sqrt{M'Q'}\left[(F')^t(D')^{-1}F'\right]\sqrt{Q'M'}
  \ge\sqrt{M'Q'}\left[(KF')^t(KD'K^t)^{-1}KF'\right]\sqrt{Q'M'}. 
\end{align}
To show the last inequality, it is sufficient to prove the following comparison
between the matrices in square brackets:
\begin{align}
  \label{eq:FDF-lowerbound} 
  (F')^t(D')^{-1}F'\ge (KF')^t(KD'K^t)^{-1}KF'. 
\end{align}
We show first
\begin{align}
  \label{eq:FDF-lowerbound-modified}
\sqrt{\W'}(F')^t(D')^{-1}F'\sqrt{\W'}\ge \sqrt{\W'}(KF')^t(KD'K^t)^{-1}KF'\sqrt{\W'}. 
\end{align}
Let $R:=F'\sqrt{\W'}$, $\tilde R:=KR$. Then, $D'=RR^t$ is invertible and
$\tilde R\tilde R^t=KRR^tK^t$ holds.  
The operators $R^t(RR^t)^{-1}R$ and $\tilde R^t(\tilde R\tilde R^t)^{-1}\tilde R$
are the orthogonal projections onto $\range R^t$ and $\range\tilde R^t$, respectively.
Since $\tilde R^t=R^tK^t$, one has $\range \tilde R^t\subseteq\range R^t$ and
consequently, $R^t(RR^t)^{-1}R\ge \tilde R^t(\tilde R\tilde R^t)^{-1}\tilde R$. 
Inserting the definitions of $R$ and $\tilde R$, this yields
\eqref{eq:FDF-lowerbound-modified}.

If $\W'$ is invertible, the inequality \eqref{eq:FDF-lowerbound} follows directly.
The set of weights $W'=(W_{e'}')_{e'\in E'}$ such that \eqref{eq:FDF-lowerbound}
holds is closed in the set of all
allowed weights. The set of all $W'$ such that $\W'$ is invertible is a
dense subset of the set of all allowed weights. Thus, \eqref{eq:FDF-lowerbound} follows.
\end{proof}

Next, we prove the auxiliary statements needed in the proof
of Lemma~\ref{le:collapse}. We use the same notation as in the lemma and
its proof.

\begin{lemma}[Auxiliary identities]
  \label{le:aux}
  Under the assumptions of Lemma~\ref{le:collapse}, the equations
  \eqref{eq:connection-prime-matrix1} and \eqref{eq:connection-prime-matrix2}
  hold. 
\end{lemma}
\begin{proof}
  We prove first \eqref{eq:connection-prime-matrix1}. This can be reformulated as
     $D_{ij}=\sum_{i'\in k^{-1}(\{i\}),\, j'\in k^{-1}(\{j\})}D_{i'j'}'$
   for all $i,j\in\Lambda$. 
For $i,j\in\Lambda$, let $D_{ij}^0$ be the contribution to $D_{ij}$ without
pinning:
\begin{align}
  D_{ij}= 1_{\{i=j\}}\W_{i\rho} + D_{ij}^0. 
\end{align}
In particular, $\sum_{j\in\Lambda}D_{ij}^0=0$ for all $i\in\Lambda$.
In the same way, we write $D_{i'j'}'= 1_{\{i'=j'\}}\W_{i'\rho}' + D_{i'j'}^{\prime 0}$. 
For the pinning part we compute
\begin{align}
  \W_{i\rho}= W_{i\rho}e^{u_i}=\sum_{i'\in k^{-1}(\{i\})}W_{i'\rho}'e^{u_i}
  =\sum_{i'\in k^{-1}(\{i\})}\W_{i'\rho}' \quad
  \text{ for all }i\in\Lambda,
\end{align}
where we used \eqref{eq:assumption-on-W-m-and-prime} and $u_{i'}'=u_i$. 
It remains to show $D_{ij}^0=\sum_{i'\in k^{-1}(\{i\}),\, j'\in k^{-1}(\{j\})}D_{i'j'}^{\prime 0}$
for all $i,j\in\Lambda$. 
For all $i$ both sides give 0 when summed over $j\in\Lambda$.
Hence, we only need to prove this identity for $i\neq j$.
For $i\neq j$, we compute
\begin{align}
  -D_{ij}^0= \W_{ij}
  =W_{ij}e^{u_i+u_j}
  =\hspace{-4mm}\sum_{\substack{i'\in k^{-1}(\{i\})\\ j'\in k^{-1}(\{j\})}}\!\! W_{i'j'}'e^{u_i+u_j}
  =\hspace{-4mm}\sum_{\substack{i'\in k^{-1}(\{i\})\\ j'\in k^{-1}(\{j\})}}\!\!\W_{i'j'}'
  =-\hspace{-4mm}\sum_{\substack{i'\in k^{-1}(\{i\})\\ j'\in k^{-1}(\{j\})}}\!\! D_{i'j'}^{\prime 0}.
\end{align}
To prove \eqref{eq:connection-prime-matrix2}, note that for $e'=\{i',j'\}\in E'$
with $k(e')=e=\{i,j\}$, one has $B_{e'}(u',s')=B_e(u,s)$.
We repeat the same argument as above 
replacing $\W_e$ and $\W_{e'}'$, respectively, with
\begin{align}
  (MQ)_e(u,s)=m_e\frac{e^{u_i+u_j}}{B_e(u,s)},\quad 
  (M'Q')_{e'}(u',s')=m_{e'}\frac{e^{u_{i'}'+u_{j'}'}}{B_{e'}(u',s')}
  =m_{e'}\frac{e^{u_i+u_j}  }{B_e(u,s)}
\end{align}
and using assumption \eqref{eq:assumption-on-W-m-and-prime}.
This concludes the proof of \eqref{eq:connection-prime-matrix2}.
\end{proof}

\begin{lemma}[Factorization of the determinant]
\label{le:factorization}   
Let $m_e\ge 0$  be parameters on the edges $e\in E$. We assume
  that the graph $G'=(\Lambda,E'=E_+\setminus\{\{i,\rho\}:i\in\Lambda\})$ obtained
  from $G_+$ by removing $\rho$ and all edges incident to it consists of $n$
  connected components $G_l'=(\Lambda_l,E_l')$, $1\le l\le n$. Let
 $G_l=(\Lambda_l\cup\{\rho\},E_l)$ denote
  the restriction of $G$ to $\Lambda_l\cup\{\rho\}$.
  \begin{enumerate}
  \item[(a)] Then, one has
    \begin{align}
      \label{eq:factor1}
    \det(\Id -M\G)=\prod_{l=1}^n\det(\Id -M_l\G_l),
  \end{align}
  where $M_l$ and $\G_l$ denote the corresponding matrices for the graph $G_l$.
  Indices $l$ with $M_l=0$ can be dropped in the product. If
  for some $\theta\in(0,1)$ one has 
  $\sqrt{M_l}\G_l\sqrt{M_l}\le(1-\theta)\Id$ for all $l$, then
  $\sqrt M\G\sqrt M\le(1-\theta)\Id$ holds as well.
\item[(b)] If for a given $l\in\{1,\ldots,n\}$ all entries of the matrix $M_l$
  vanish with the possible exception of a single non-zero entry $m_{e_l}$
  with $e_l\in E_l$, then
  \begin{align}
    \label{eq:factor2}
    \det(\Id -M_l\G_l)=1-m_{e_l}(\G_l)_{e_l,e_l}. 
  \end{align}
  Moreover, if $m_{e_l}(\G_l)_{e_l,e_l}\le 1-\theta$ for some $\theta\in(0,1)$, then $\sqrt{M_l}\G_l\sqrt{M_l}\le(1-\theta)\Id$.
  \end{enumerate}
\end{lemma}
\begin{proof}
  The matrix $D$ and consequently also its inverse
  are block diagonal with blocks indexed by $\Lambda_l$, $1\le l\le n$, i.e., it
  satisfies $D_{ij}=0$ for $i\in\Lambda_k$ and $j\in\Lambda_l$
  with $k\neq l$. Thus, for edges $e\in E_k$, $f\in E_l$, one has
  $(F^tD^{-1}F)_{ef}=\sum_{i,j\in\Lambda}(1_{\{ e_+=i\}}-1_{\{ e_-=i\}})(1_{\{ f_+=j\}}-1_{\{ f_-=j\}}) D_{ij}^{-1}=0$ because $e\cap f\subseteq\{\rho\}$. This shows that
  $M\G=M\sqrt QF^tD^{-1}F\sqrt Q$ is block diagonal with blocks indexed by $E_l$,
  $1\le l\le n$.
  Claim (a) follows.

  For part (b), we write $\det(\Id -M_l\G_l)=\det(\Id -\sqrt{M_l}\G_l\sqrt{M_l})$
  and observe that the matrix $\sqrt{M_l}\G_l\sqrt{M_l}$ has at most one non-zero entry
  $m_{e_l}(\G_l)_{e_l,e_l}$. 
\end{proof}

\begin{remark}
  \label{rem:resistance}
In the situation of part (b) of Lemma~\ref{le:factorization},
let $D_l$ denote the weighted Laplacian on $G_l$ as defined
in \eqref{eq:representation-D-with-pinning-old}. We set $e_l=\{i_0,j_0\}$ and
\begin{align}
  \label{eq:def-c-e}
  c_e:=\frac{\W_e}{Q_{\{i_0,j_0\}}}
  =\frac{W_{ij}e^{u_i+u_j}B_{i_0j_0}}{e^{u_{i_0}+u_{j_0}}}, \quad e=\{i,j\}\in E_l, 
\end{align}
cf.\ \eqref{eq:def-c-example}. 
Let $\cR^{G_l}(c,i_0\leftrightarrow j_0)$ be the effective resistance between $i_0$ and $j_0$
in the electrical network consisting of the graph $G_l$ endowed with the conductances
$c_e$, $e\in E_l$. 
Then, by \eqref{eq:G-ee-as-resistance-ex} we have
\begin{align}
  (\G_l)_{e_l,e_l}=&Q_{e_l}(1_{i_0}-1_{j_0})^t D^{-1}_l(1_{i_0}-1_{j_0})
                     =\cR^{G_l}(c,i_0\leftrightarrow j_0),
\end{align}
recall that $1_i=(\delta_{ij})_{j\in\Lambda_l}$. By Rayleigh's monotonicity
principle, the effective resistance increases when removing edges. In particular,
for $\rho\not\in e_l$, one has 
\begin{align}
  (\G_l)_{e_l,e_l}=\cR^{G_l}(c,i_0\leftrightarrow j_0)\le\cR^{G_l'}(c,i_0\leftrightarrow j_0).
     \label{eq:G-resistance}
\end{align}
\end{remark}
In the next theorem, we combine the results from the current subsection into
a theorem which will be used for concrete applications below. 

Given an edge set $\tilde E$, let $\Lambda_{\tilde E}$ denote the set of all vertices
incident to at least one edge in $\tilde E$. In other words, $\Lambda_{\tilde E}$ is the
union of all elements in $\tilde E$. We call the edge set $\tilde E$ connected
if the graph $(\Lambda_{\tilde E},\tilde E)$ is connected.  

\begin{theorem}[Summary of this subsection]
  \label{thm:summary-section}
Let  $G=(\Lambda\cup\{\rho\},E)$ be the complete graph with
weights $W_e\ge 0$, $e\in E$, such that $G_+$ is connected, and 
let
$m_e\ge 0$, $e\in E$, be parameters. Let $e_1=\{x_1,y_1\},\ldots,e_n=\{x_n,y_n\}$
be the edges in $E$ with $m_e>0$. Assume that for $1\le l\le n$, we have pairwise disjoint connected subsets
$E_l$ of $E_+$ such that $x_l,y_l\in \Lambda_{E_l}$. 
Note that $e_l$ need not be an element of $E_l$ and that the sets $\Lambda_{E_l}$
need not be pairwise disjoint for $l=1,\ldots,n$. 
For $1\le l\le n$, let $\cR^{G_l}(c,x_l\leftrightarrow y_l)$ denote the effective
resistance between $x_l$ and $y_l$
in the electrical network consisting of the graph
$G_l=(\Lambda_{E_l},E_l)$ endowed with the
conductances $c_e$, $e\in E_l$, defined in \eqref{eq:def-c-e} with $\{i_0,j_0\}=\{x_l,y_l\}$. 
Assume that we have parameters $\delta_e\in(0,\infty]$, $e\in E$, such that
on the event $U$ defined in \eqref{eq:def-U} one has 
$m_{e_l}\cR^{G_l}(c,x_l\leftrightarrow y_l)<1$ for all $1\le l\le n$.

Then Assumption \ref{ass:m-W} is valid and on the event $U$, one has 
\begin{align}
  \label{eq:estimate-det-summary}
\det(\Id-M\G)\ge\prod_{l=1}^n(1-m_{e_l}\cR^{G_l}(c,x_l\leftrightarrow y_l)).
\end{align}
Whenever $e_l$ is not a pinning edge, i.e.\ $\rho\not\in e_l$, the term
$\cR^{G_l}(c,x_l\leftrightarrow y_l)$ in the last product may be replaced by
$\cR^{G_l'}(c,x_l\leftrightarrow y_l)$, where $G_l'$ is obtained from $G_l$
by removing $\rho$ and all edges incident to it. 
Moreover, $\cR^{G_l}(c,x_l\leftrightarrow y_l)$ is monotonically decreasing
in the conductances $c$: if $c'\ge c\ge 0$ componentwise, then
\begin{align}
  \label{eq:rayleigh}
  \cR^{G_l}(c',x_l\leftrightarrow y_l)\le \cR^{G_l}(c,x_l\leftrightarrow y_l).
\end{align}
\end{theorem}
\begin{proof}
  Note that the claim \eqref{eq:estimate-det-summary} is a pointwise statement
  on the event $U$; therefore we prove it for a given sample point in $U$. 
 Take $\theta>0$ so small that $\theta\le 1-m_{e_l}\cR^{G_l}(c,x_l\leftrightarrow y_l)$
    for all $1\le l\le n$.
  Adding a small positive number $\eps$ to all pinning weights $W_{i\rho}$,
  $i\in\Lambda$, it suffices to prove the theorem under the additional
  assumption that all pinning weights are positive: $W_{i\rho}>0$.
  Indeed, the inequality 
    $\theta\le 1-m_{e_l}\cR^{G_l}(c,x_l\leftrightarrow y_l)$
    for $\eps=0$ implies the same inequality for $\eps>0$ by Rayleigh's
    monotonicity principle.
    Moreover, both sides of \eqref{eq:estimate-det-summary} are
    continuous functions of the additional summand $\eps\ge 0$, which implies
    that \eqref{eq:estimate-det-summary} holds for $\eps=0$ given that it
    holds for small $\eps>0$.
   By the monotonicity lemma \ref{le:monotonicity}, it is sufficient to prove
  Assumption \ref{ass:m-W} and the inequality
  \eqref{eq:estimate-det-summary} for the modified model obtained by taking $W_e=0$ 
  for all edges $e\in E$ which are in none of the sets 
  $E_l$, $1\le l\le n$. Note that the right-hand side of \eqref{eq:estimate-det-summary}
  does not depend on the weights we set to 0 and 
  that the modified model remains well defined
  because of the positive pinning weight of all vertices.

  In order to apply Lemma~\ref{le:collapse}, we separate vertices as follows:
  We consider the disjoint union $\Lambda'$ of the 
  sets $\Lambda_{E_l}\setminus\{\rho\}$ together with all vertices $i\in\Lambda$ which are in  none of these
  sets. More formally, let $i^l$ denote the copy of $i\in\Lambda_{E_l}\setminus\{\rho\}$ corresponding to the
  subgraph $(\Lambda_{E_l},E_l)$. Then, one has
\begin{align}
  \Lambda'=\{i^l:1\le l\le n, i\in \Lambda_{E_l}\setminus\{\rho\}\}\cup
  \{i\in\Lambda:i\notin \Lambda_{E_l}\text{ for all }1\le l\le n\}. 
\end{align}  
Let $G'=(\Lambda'\cup\{\rho\},E')$ be the corresponding complete graph and
$k:\Lambda'\cup\{\rho\}\to\Lambda\cup\{\rho\}$ be the natural map
$i^l\mapsto i$ for $i\in\Lambda_{E_l}\setminus\{\rho\}$, $i\mapsto i$ for vertices not contained in
any $\Lambda_{E_l}\setminus\{\rho\}$, and $\rho\mapsto\rho$, cf.\ Fig.\ \ref{fig:collapse}.
 \begin{figure}
 \centerline{ \includegraphics[width=0.9\textwidth]{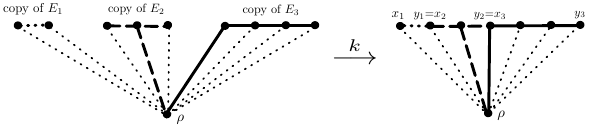}}
  \caption{Illustration of the map $k$}
  \label{fig:collapse}
\end{figure}
  We endow the edges $e'\in E'$ with weights
  $W_{e'}'\ge 0$ as follows. For any edge $\{i,j\}\in E_l$ with $1\le l\le n$ and
  $\rho\notin\{i,j\}$,
  the corresponding edge $\{i^l,j^l\}$ inherits its weight: $W_{i^lj^l}'=W_{ij}$.
  Observe that a vertex $i\in\Lambda$ can be incident to edges in several $E_l$'s.
  For pinning edges $\{i,\rho\}$, we split the weight $W_{i\rho}$ as follows.
  Given $i\in\Lambda$, set $I_i=\{l\in\{1,\ldots,n\}:\; i\in\Lambda_{E_l}\}$. In the case $I_i\neq\emptyset$, i.e., $\{i,\rho\}\in E_l$ for some $l$,
    we set 
    $W_{i^j\rho}'=(W_{i\rho}-\varepsilon)\delta_{jl}+\varepsilon/|I_i|$ for
    $j\in I_i$; the index $l$ is unique because the $E_l$ are
    pairwise disjoint. In the case $I_i=\emptyset$, i.e., $\{i,\rho\}\notin E_l$ for
    all $l$, we set $W_{i\rho}'=W_{i\rho}$. This construction guarantees 
    $\sum_{i'\in k^{-1}(\{i\})}W_{i'\rho}'=W_{i\rho}$.  Note that $W_{i'\rho}'\ge\eps/n>0$
    for all $i'\in\Lambda'$ because $W_{i\rho}\ge\eps$ for all $i\in\Lambda$. 
  For all other edges $e'\in E'$, we set $W_{e'}'=0$.
  For $1\le l\le n$, we set $e_l'=\{x_l^l,y_l^l\}\in E'$ and $m_{e_l'}'=m_{e_l}$, but
  $m_{e'}'=0$ for all other edges $e'\in E'$. Then, Lemmas \ref{le:collapse} and
  \ref{le:factorization} and Remark~\ref{rem:resistance} are applicable.
  In particular, on $U$ the bound $m_{e_l}\cR^{G_l}(c,x_l\leftrightarrow y_l)\le 1-\theta$ and
  Remark~\ref{rem:resistance} imply 
  $m_{e_l}(\G_l)_{e_l,e_l}\le 1-\theta$ and then, by Lemma~\ref{le:factorization}(b),
  $\sqrt{M_l}\G_l\sqrt{M_l}\le(1-\theta)\Id$ for all $l$. 
  As a consequence, by Lemma~\ref{le:factorization}(a),
  $\sqrt{M'}\G'\sqrt{M'}\le(1-\theta)\Id$, which implies
  $\sqrt{M}\G\sqrt{M}\le(1-\theta)\Id$ by
  Lemma~\ref{le:collapse}. Note that this estimate is uniform in $\varepsilon>0$. 
  Thus, by Lemma~\ref{le:equivalence-of-hypotheses}, Assumption \ref{ass:m-W} is satisfied,
  also in the limit $\varepsilon\to 0$.
  Combining formulas \eqref{eq:claim-gluing}, \eqref{eq:factor1}, \eqref{eq:factor2},
  and \eqref{eq:G-resistance}, we obtain
  \begin{align}
    \det(\Id-M\G(u,s))\ge&\det(\Id-M'\G'(u',s'))
    =\prod_{l=1}^n\det(\Id -M_l\G_l)\nonumber\\
    =&\prod_{l=1}^n(1-m_{e_l}(\G_l)_{e_l,e_l})
       \ge \prod_{l=1}^n(1-m_{e_l}\cR^{G_l}(c,x_l\leftrightarrow y_l)).
       \label{eq:det-lower-bound-in-proof}
  \end{align}
  Rayleigh's monotonicity principle and its consequence \eqref{eq:G-resistance} allow
us to conclude. 
\end{proof}

\subsection{Proof of Theorem~\ref{thm:one-dim}}
\label{sec:boundingoned}

In this section, we will prove Theorem~\ref{thm:one-dim}.
We consider the model described in Section \ref{sec:onedimensionaleff}
with $\alpha>3$, possibly extended by additional non-nearest
neighbor edges~$e$ with weights $W_e\ge 0$.
Throughout this section, $N$ is fixed and hence we drop the dependence on $N$ in the
notation. The important point is that the final estimates are uniform in $N$. 

Recall the notation $e =\{x_e,y_e\}$ with $x_e<y_e$
for any edge $e$. We take a sequence
of cutoff parameters $\delta_l$, $l\ge 2$, and 
associate the $\delta_l$ to nearest-neighbor edges on $\N$ by setting 
$\delta_{\{l-1,l\}}:=\delta_l$ for $l\ge 2$. 
Remember from \eqref{eq:introduction-chi-e}
that we defined $\chi_{ij}=\chi_e=1_{(-\infty,1+\delta_e)}$ for $e=\{i,j\}$.
For $i\le j$, we set 
\begin{align}
  \overline\chi_{ij}:=\prod_{l=i+1}^j\chi_{l-1,l}(B_{l-1,l}). 
\end{align}
Note that in the case $i=j$ this is the empty product $\overline\chi_{ii}=1$.

\begin{lemma}
  \label{le:exp-one-product}
  Let  $E',E''$ be two disjoint subsets of the edge set of the complete
  graph over $[N]$, with $E'\subseteq E_N$, i.e., every $e\in E'$
  is a nearest neighbor edge: $y_e-x_e=1$. Assume that for any two
  different $e,\tilde e\in E'\cup E''$ the open intervals
  $(x_e,y_e)$ and $(x_{\tilde e},y_{\tilde e})$ are disjoint. Assume further that 
  $\sum_{l=2}^\infty\sqrt{\delta_l}<\infty$ and set $C_\delta=e^{\sqrt 2\sum_{l=2}^\infty\sqrt{\delta_l}}$,
  which is uniform in $N$. 
  Then, for all exponents $m_e$ satisfying $0\le m_e<W_e$ for $e\in E'$ and 
  $0\le m_e< \left(C_\delta\sum_{i=x_e+1}^{y_e}\frac{1}{W_{i-1,i}}\right)^{-1}$ for $e\in E''$, 
  we have 
   \begin{align}
    \E_{W,h}^{[N]}\left[\prod_{e\in E'}B_e^{m_e}\cdot\prod_{e\in E''}(B_e^{m_e}\overline\chi_e)
    \right]\le\prod_{e\in E'}\frac{1}{1-\frac{m_e}{W_e}}
    \prod_{e\in E''}\frac{1}{1-C_\delta m_e\sum_{i=x_e+1}^{y_e}\frac{1}{W_{i-1,i}}}.
  \end{align}
\end{lemma}

Note that $\sum_{i=x_e+1}^{y_e}\frac{1}{W_{i-1,i}}$ is the effective resistance
of a network consisting of a series of conductances $W_{i-1,i}$.

\medskip\noindent
\begin{proof}[Proof of Lemma \ref{le:exp-one-product}]
  We will show that Assumption \ref{ass:m-W} is satisfied and
  that on the support of $\prod_{e\in E''}\overline\chi_e$, we have 
  \begin{align}
    \label{eq:lower-bound-det-id-mg}
    \det(\Id-M\G)\ge \prod_{e\in E'}\left(1-\frac{m_e}{W_e}\right)
    \prod_{e\in E''}\left(1-C_\delta m_e\sum_{i=x_e+1}^{y_e}\frac{1}{W_{i-1,i}}\right).
  \end{align}
  It suffices to prove this in the case $W_{ij}=0$ for all edges
  $\{i,j\}$ with $|i-j|\ge 2$. The general case follows from Lemma
  \ref{le:monotonicity}.   
 Using \eqref{eq:lower-bound-det-id-mg} and Lemma~\ref{le:upper-bound-prod-B-chi-det}, we argue 
   \begin{align}
   \E_{W,h}^{[N]}&\left[\prod_{e\in E'}B_e^{m_e}\cdot\prod_{e\in E''}(B_e^{m_e}\overline\chi_e)
     \right]
   =\E_{W,h}^{[N]}\left[\prod_{e\in E'}B_e^{m_e}\cdot\prod_{e\in E''}(B_e^{m_e}\overline\chi_e)
      \frac{\det(\Id-M\G)}{\det(\Id-M\G)} \right]\nonumber\\
\le& \frac{\E_{W,h}^{[N]}\left[\prod_{e\in E'}B_e^{m_e}\cdot\prod_{e\in E''}(B_e^{m_e}\overline\chi_e)
      \det(\Id-M\G) \right]}{ \prod_{e\in E'}\left(1-\frac{m_e}{W_e}\right)
     \prod_{e\in E''}\left(1-C_\delta m_e\sum_{i=x_e+1}^{y_e}\frac{1}{W_{i-1,i}}\right)}\nonumber\\
\le& \prod_{e\in E'}\frac{1}{1-\frac{m_e}{W_e}}
    \prod_{e\in E''}\frac{1}{1-C_\delta m_e\sum_{i=x_e+1}^{y_e}\frac{1}{W_{i-1,i}}},
   \end{align}
   which concludes the proof.
   
   To prove \eqref{eq:lower-bound-det-id-mg}, we apply Theorem~\ref{thm:summary-section}
   as follows. For $e\in E'\cup E''$, we consider the pairwise disjoint sets of edges 
   $E_e=\{\{ i-1,i\}:\, x_e<i\le y_e\}$ and the corresponding graphs
   $G_e=(\{x_e,\ldots,y_e\},E_e)$. In particular, $E_e=\{e\}$ for $e\in E'$. 
   We claim that 
   \begin{align}
     \label{eq:hypothesis-theorem-summary}
     m_e\cR^{G_e}(c,x_e\leftrightarrow y_e)<1 \quad\text{ for }e\in E'\cup E''
   \end{align}
   holds on $U\subseteq\supp\prod_{e\in E''}\overline\chi_e$, cf.\ definition
   \eqref{eq:def-U} of $U$. 
   Indeed, for the nearest-neighbor edges $e\in E'$, we have
   \begin{align}
     \label{eq:est-R1}
 m_e\cR^{G_e}(c,x_e\leftrightarrow y_e)=&\frac{m_e}{c_e}
  =\frac{m_ee^{u_{x_e}+u_{y_e}}}{W_ee^{u_{x_e}+u_{y_e}}B_e}\le\frac{m_e}{W_e}<1.
\end{align}
For edges $e\in E''$, we will show
\begin{align}
  \label{eq:est-R2}
  m_e\cR^{G_e}(c,x_e\leftrightarrow y_e)\le m_eC_\delta \sum_{i=x_e+1}^{y_e}\frac{1}{W_{i-1,i}}<1.
\end{align}
The inequalities ``$<1$'' in \eqref{eq:est-R1} and  \eqref{eq:est-R2} follow from
the assumptions on $m_e$.
Therefore, by Theorem~\ref{thm:summary-section}, Assumption \ref{ass:m-W} is valid
and we have
   \begin{align}
     \det(\Id-M\G) \ge \prod_{e\in E'\cup E''}(1-m_e\cR^{G_e}(c,x_e\leftrightarrow y_e)). 
   \end{align}
   Inserting the bounds \eqref{eq:est-R1} and  \eqref{eq:est-R2}, we obtain
   \eqref{eq:lower-bound-det-id-mg}.

   Finally, we prove \eqref{eq:est-R2}. Consider an edge $e\in E''$.
Using the definition \eqref{eq:def-c-e} of the conductances $c_e$ and
the series law for the resistances between $x_e$ and $y_e$, we obtain 
\begin{align}
\cR^{G_e}(c,x_e\leftrightarrow y_e)
  &= \sum_{i=x_e+1}^{y_e}\frac{1}{c_{i-1,i}}
     =\sum_{i=x_e+1}^{y_e}\frac{e^{u_{x_e}+u_{y_e}}}{W_{i-1,i}e^{u_{i-1}+u_i}B_e}\nonumber\\
     &\le \sum_{i=x_e+1}^{y_e}\frac{e^{u_{x_e}-u_{i-1}+u_{y_e}-u_i}}{W_{i-1,i}},
     \label{eq:upper-bound-R-j+1-y-new}
\end{align}
where we used $B_e\ge 1$. On the support of $\overline\chi_e$, one has 
for all $l\in\{x_e+1,\ldots,y_e\}$, 
\begin{align}
1+\delta_l>B_{l-1,l}\ge\cosh(u_{l-1}-u_l)\ge 1+\frac12(u_{l-1}-u_l)^2  
\end{align}
and hence $|u_{l-1}-u_l|\le\sqrt{2\delta_l}$. For $i=x_e+1,\ldots,y_e$, we estimate 
\begin{align}
  u_{x_e}-u_{i-1}+u_{y_e}-u_i\le\sum_{\substack{l=x_e+1\\l\neq i}}^{y_e}|u_{l-1}-u_l|
  \le \sqrt 2\sum_{\substack{l=x_e+1\\l\neq i}}^{y_e}\sqrt{\delta_l}
  \le \log C_\delta .
  \label{eq:diff-u-s-le-C-delta}
\end{align}
Inserting this in \eqref{eq:upper-bound-R-j+1-y-new} concludes the proof of
\eqref{eq:est-R2}. 
\end{proof}

Using Lemma~\ref{le:exp-one-product}, we now prove the bounds
  in the inhomogeneous one-dimensional model. 

\medskip\noindent
\begin{proof}[Proof of Theorem~\ref{thm:one-dim}]
  We prove part 2 of the theorem first, starting
  with a pointwise bound on each
  factor $B_e^{m_e}$ corresponding to an edge
    $e\in E''$. For this purpose, 
  keeping $e$ fixed for the moment, we
  abbreviate $x=x_e$, $y=y_e\ge x_e+1$ and
  $m=m_e\ge 0$.
  We consider the partition of unity (cf.\ Fig.~\ref{fig:onedpartition})
  %%%%%%%%%%%%%%%%%%%%%%%%%%%%%%%%%%
  \begin{figure}
 \centerline{ \includegraphics[width=7cm]{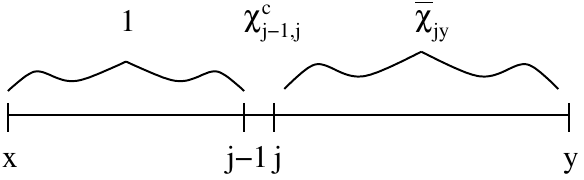}}
  \caption{Cutoff functions used in the partition of unity}
  \label{fig:onedpartition}
\end{figure}
%%%%%%%%%%%%%%%%%%%%%%%%%%%%%%%%%%%%%%%%%
  \begin{align}
    &1=\overline\chi_{xy}+(1-\overline\chi_{xy})
    =\overline\chi_{xy}+\sum_{j=x+1}^y\chi^c_{j-1,j}\overline\chi_{jy}
     \text{ with }
    \chi^c_{j-1,j}:=1-\chi_{j-1,j}(B_{j-1,j}), 
  \end{align}
  where the constants $\delta_j$ appearing in the definition
  $\chi_{j-1,j}(B)=1_{\{B< 1+\delta_j\}}$ will be specified later. 
  This means that
  either all nearest-neighbor pairs between $x$ and $y$ are protected
  in the sense that the factor $\chi_{j-1,j}$ of the corresponding edge
  is present or there is a largest index $j$ for which the edge
  $\{j-1,j\}$ is unprotected. We
  estimate the $j$-th summand in the last sum, using the Chebyshev-type
  inequality
  \begin{align}
    \label{eq:chebyshev-type-bound-chi}
    \chi_{j-1,j}^c=1_{\{B_{j-1,j}\ge 1+\delta_j\}}
    \le \left(\frac{B_{j-1,j}}{1+\delta_j}\right)^{p_j}, 
  \end{align}
  cf.\ \eqref{eq:chebyshev-type-inequ}, where $p_j>0$ will also be chosen later.
  Inserting this bound and multiplying by $B_e^m=B_{xy}^m$ yields
  \begin{align}
    \label{eq:bound-Bem}
    &B_e^m\le B_{xy}^m\overline\chi_{xy}
      +\sum_{j=x+1}^y(1+\delta_j)^{-p_j}B_{xy}^mB_{j-1,j}^{p_j}
      \overline\chi_{jy}.
  \end{align}
  We estimate $B_{xy}$ in the $j$-th summand as follows.
  Recall from \eqref{eq:triangle-strat} that $B_{ij}\le 2B_{ik}B_{kj}$ 
  holds for all $i,j,k$.
    Given $j$ with $x+1\le j\le y$, we iterate it 
  $j-x$ times. This yields 
  \begin{align}
    B_{xy}\le 2^{j-x}B_{jy}\prod_{i=x+1}^jB_{i-1,i},
    \label{eq:triangle-iterated}
  \end{align}
  with the convention $B_{yy}=1$.
   Inserting \eqref{eq:triangle-iterated} in the $j$-th summand in \eqref{eq:bound-Bem}, we obtain
  \begin{align}
    \label{eq:j-th-summand-new}
    B_{xy}^m
    \le B_{xy}^m\overline\chi_{xy}
    +\sum_{j=x+1}^y
    \frac{2^{(j-x)m}}{(1+\delta_j)^{p_j}}
    \left(\prod_{i=x+1}^{j-1}B_{i-1,i}^m\right)B_{j-1,j}^{m+p_j}B_{jy}^m
    \overline\chi_{jy},
  \end{align}
  cf.\ \eqref{eq:split-B_e^m2}. 
  We now take the product of both sides of this inequality
  over $e=\{x_e,y_e\}\in E''$,
  not suppressing the $e$-dependence in the notation anymore,
  and expand the product of sums on the r.h.s.\ into a sum of products.
  To keep the calculations readable, we introduce some abbreviations,
  locally for this proof.
  Let $X_{e,j}= \frac{2^{(j-x)m}}{(1+\delta_j)^{p_j}}
    \left(\prod_{i=x+1}^{j-1}B_{i-1,i}^m\right)B_{j-1,j}^{m+p_j}B_{jy}^m
    \overline\chi_{jy}$, $j\in\{x_e+1,\ldots, y_e\}$,
  denote the summand on the r.h.s.\ of \eqref{eq:j-th-summand-new}
  indexed by $j$, and $X_{e,x_e}=B_{x_ey_e}^{m_e}\overline\chi_{x_ey_e}$ 
  the first summand.
  Using these abbreviations, we rewrite
  \eqref{eq:j-th-summand-new} in the form 
  \begin{align}
    \label{eq:j-th-summand-new-X}
    B_{x_ey_e}^{m_e}\le \sum_{j=x_e}^{y_e}X_{e,j}
  \end{align}
  for $e\in E''$.
  Taking the product over all such edges $e$ and expanding
  the resulting product of sums in a sum of products, 
  we obtain
  \begin{align}
    \label{eq:prod-B-sumprod-X}
    \prod_{e\in E''}B_{x_ey_e}^{m_e}\le
    \prod_{e\in E''}\sum_{j=x_e}^{y_e}X_{e,j}
    =\sum_{(j_e)_{e\in E''}\in\cI}\prod_{e\in E''}X_{e,j_e}
  \end{align}
  with the Cartesian product
  $\cI=\prod_{e\in E''}\{x_e,\ldots, y_e\}$.
  Therefore,
  \begin{align}
 \E\left[\prod_{e\in E'}B_{x_ey_e}^{m_e}\cdot\prod_{e\in E''}B_{x_ey_e}^{m_e}\right]
    \le\sum_{(j_e)_{e\in E''}\in\cI}\E\left[\prod_{e\in E'}B_{x_ey_e}^{m_e}\cdot\prod_{e\in E''}X_{e,j_e}\right]. 
  \end{align}
  We will show that Lemma~\ref{le:exp-one-product} is applicable
    for an appropriate choice of $\delta_j$ and $p_j$, and then use it to
  estimate the expectations in the sum.
  
  \paragraph{Illustration of the strategy in two special cases.}
  The following illustration is intended to clarify the strategy,
  although it is not a part of the formal proof.
  \begin{itemize}
    \item{\bf  First special case}:
      Assume that $E''=\emptyset$ and $E'$ consists
      of a single edge $e=\{x_e,y_e\}$. By \eqref{eq:cite-from-dmr23} or alternatively
      Lemma~\ref{le:exp-one-product}, one has
  \begin{align}
     \label{eq:def-Y-e}
    \E\left[B_{x_ey_e}^{m_e}\right]\le \frac{1}{1-\frac{m_e}{W_e}}=:Y_e
  \end{align}
  because $m_e<W_e$.
\item{\bf  Second special case}: Assume that 
  $E'=\emptyset$, $E''$ consists of a single edge $e=\{x_e,y_e\}$, and
  Lemma~\ref{le:exp-one-product} is applicable in the following estimates.
  It gives us 
  \begin{align}
    \label{eq:def-Y-ex}
    \E\left[X_{e,x_e}\right]= \E\left[B_{x_ey_e}^{m_e}\overline\chi_{x_ey_e}\right]
    \le \frac{1}{1-C_\delta m_e\sum_{k=x_e+1}^{y_e}\frac{1}{W_{k-1,k}}}=:Y_{e,x_e}
  \end{align}
  and for $j\in\{x_e+1,\ldots, y_e\}$
  \begin{align}
    &\E\left[X_{e,j}\right]= \frac{2^{(j-x_e)m_e}}{(1+\delta_j)^{p_j}}
    \E\left[\left(\prod_{i=x_e+1}^{j-1}B_{i-1,i}^{m_e}\right)B_{j-1,j}^{m_e+p_j}B_{jy_e}^{m_e}
      \overline\chi_{jy_e}\right]
        \label{eq:def-Y-e-j}\\
    &\le \frac{2^{(j-x_e)m_e}}{(1+\delta_j)^{p_j}}
    \prod_{i=x_e+1}^{j-1}\frac{1}{1-\frac{m_e}{W_{i-1,i}}}
    \cdot \frac{1}{1-\frac{m_e+p_j}{W_{j-1,j}}}
                 \cdot\frac{1}{1-C_\delta m_e\sum_{k=j+1}^{y_e}\frac{1}{W_{k-1,k}}}
         =:Y_{e,j}.\nonumber
  \end{align}
\end{itemize}

\paragraph{General case.}
Let us now continue
with the formal proof in the general case, which encompasses the two special
cases, but involves some additional combinatorial effort.
Assuming that Lemma~\ref{le:exp-one-product} is applicable, which
we check below, it will give 
  \begin{align}
    \E\left[\prod_{e\in E'}B_{x_ey_e}^{m_e}\cdot \prod_{e\in E''}X_{e,j_e}\right]
    \le\prod_{e\in E'}Y_e\cdot \prod_{e\in E''}Y_{e,j_e}
  \end{align}
  for all $(j_e)_{e\in E''}\in\cI$ with  $Y_e$ defined in \eqref{eq:def-Y-e}
  and $Y_{e,j_e}$ defined in \eqref{eq:def-Y-ex} for $j_e=x_e$
  and in \eqref{eq:def-Y-e-j} for $j_e>x_e$.
  Summing over $\cI$ and recombining the sum over products into a product of
  sums, we obtain
  \begin{align}
    \label{eq:E-prod-B-prod-Y}
    \E\left[\prod_{e\in E'\cup E''}B_{x_ey_e}^{m_e}\right]\le
    \sum_{(j_e)_{e\in E''}\in\cI}\prod_{e\in E'}Y_e\cdot \prod_{e\in E''}Y_{e,j_e}
    =\prod_{e\in E'}Y_e\cdot \prod_{e\in E''}\sum_{j=x_e}^{y_e}Y_{e,j}.
  \end{align}
  \paragraph{Parameter tuning.}
  To check applicability of Lemma~\ref{le:exp-one-product}, we need $\delta_j$ and $p_j$
  such that $\sum_{l=2}^\infty\sqrt{\delta_l}<\infty$ and all denominators
  in the $Y$-factors are strictly positive. For later use, some estimates
  will be sharper than needed at this point. 

  Recall the assumption on the weights $W_{j-1,j}\ge\overline Wj^\alpha$ with
  $\alpha>3$. Let $\kappa\in(0,1]$ and $\gamma\ge 0$ with $\alpha-\gamma>3$
  be given. We set $\delta_j:=j^{-(\alpha-\gamma-1)}$ for $j\ge 2$.
  With this choice, we have 
  \begin{align}
    \label{eq:bound-sum-delta}
   & \sum_{l=2}^\infty\sqrt{\delta_l}
  =\sum_{l=2}^\infty l^{-\frac{\alpha-\gamma-1}{2}}\le\int_1^\infty t^{-\frac{\alpha-\gamma-1}{2}}\, dt
  =\frac{2}{\alpha-\gamma-3}<\infty.
  \end{align}
  Note that the assumption $\alpha>3$ is necessary for the existence of $\gamma\ge0$
  with $\alpha-\gamma>3$, which is essential for having summability of the
  $\sqrt{\delta_l}$.

  In Lemma~\ref{le:thm-quantitative}, below, we collect a choice of parameters
  for which this proof works. We now take $\cdrei,\cvier,\wnull$ from there, i.e.,
  as in \eqref{eq:def-W0} and \eqref{eq:def-czwei}, and
    $\overline W\ge\wnull$. For a given $e=\{x,y\}\in E''$,
    we consider $Y_{e,j}$ with $x\le j\le y$,
    suppressing the $e$-dependence again in the notation.
    We need to estimate $m/W_{i-1,i}$ for $x<i<y$, $(m+p_j)/W_{j-1,j}$ for $x<j\le y$,
    and $C_\delta m\sum_{k=j+1}^y\frac{1}{W_{k-1,k}}$ for $x\le j\le y$. 
Recall that
    $m_x=\cdrei\overline W x^\gamma= \cvier\overline W x^\gamma/\log 2$
    and let $m\in[0,m_x]$. 
For all $i>x$, using 
$\gamma-\alpha<0$ and $\cdrei\le(16\log 2)^{-1}\le\frac14$, we estimate
\begin{align}
  \label{eq:m-W-14}
\quad & \frac{m}{W_{i-1,i}}\le \frac{m_x}{\overline Wi^\alpha}
 =\cdrei x^\gamma i^{-\alpha}
  \le\cdrei x^{\gamma-\alpha}\le \cdrei\le\frac14.
\end{align}
For $j\in\N$, we set 
$p_j:=4\cvier \overline Wj^\alpha$.
For $x<j\le y$, we obtain with $\cvier\le\frac{1}{16}$ the bound 
\begin{align}
  &  \frac{m+p_j}{W_{j-1,j}}
    \le \frac14+ \frac{p_j}{W_{j-1,j}}
    \le \frac14+\frac{p_j}{\overline Wj^\alpha}=\frac14+4\cvier\le\frac12.
\label{eq:est-m-p-Winv}
\end{align}
Finally, for $x\le j\le y$, we estimate, using $\alpha>3$, 
\begin{align}
  \sum_{k=j+1}^y\frac{1}{W_{k-1,k}}
  \le\sum_{k=j+1}^y\frac{1}{\overline Wk^\alpha}
  \le\int_x^\infty \frac{1}{\overline Wt^\alpha}\, dt
  =\frac{x^{1-\alpha}}{\overline W(\alpha-1)}\le\frac{x^{1-\alpha}}{2\overline W} .
  \label{eq:bound-sum-1-over-W}
\end{align}
Because of \eqref{eq:bound-sum-delta}, the constant $C_\delta$ from
Lemma~\ref{le:exp-one-product} satisfies
$1\le C_\delta=e^{\sqrt 2\sum_{l=2}^\infty\sqrt{\delta_l}}
  \le e^{\frac{2\sqrt 2}{\alpha-\gamma-3}}$.
  Hence, our choice of $\cdrei$ in \eqref{eq:def-W0} and \eqref{eq:def-czwei}
yields $\cdrei\le\frac23\kappa/C_\delta$. 
Together with \eqref{eq:bound-sum-1-over-W}, using $\gamma+1-\alpha<0$, it follows 
\begin{align}
  \label{eq:est-C-m-sum}
  & C_\delta m\sum_{k=j+1}^y\frac{1}{W_{k-1,k}}\le C_\delta m_x \frac{x^{1-\alpha}}{2\overline W}
  =\frac{C_\delta\cdrei x^{\gamma+1-\alpha}}{2}
  \le\frac{C_\delta\cdrei}{2} \le \frac{\kappa}{3}<1.
\end{align}
This shows that the assumptions
of Lemma~\ref{le:exp-one-product} needed above are fulfilled.

We now use the estimates above to bound the $Y$-variables defined
in \eqref{eq:def-Y-ex} and \eqref{eq:def-Y-e-j}. 
Using $\kappa\in(0,1]$, one has $(1-t)^{-1}\le 1+\kappa/2$ for all $t\in[0,\kappa/3]$.
Consequently, \eqref{eq:est-C-m-sum} implies
\begin{align}
  \label{eq:est-1-over-1-C-delta}
  &  \frac{1}{1-C_\delta m\sum_{k=j+1}^y\frac{1}{W_{k-1,k}}}
  \le 1+\frac{\kappa}{2}\le 2\text{ for }x\le j\le y. 
\end{align}
In the case $j=x$, this yields 
\begin{align}
  \label{eq:est-Y-e-x}
  Y_{e,x}\le 1+\frac{\kappa}{2}. 
\end{align}
We consider now $x<j\le y$. For $t\in[0,\frac12]$, we have 
$(1-t)^{-1}\le\min\{2,e^{2t}\}$. 
We apply this estimate to  $t=\frac{m}{W_{i-1,i}}\le\frac14$,
cf.\ \eqref{eq:m-W-14},
with $x<i<j$, and to $t=\frac{m+p_j}{W_{j-1,j}}\le\frac12$,
cf.\ \eqref{eq:est-m-p-Winv}, to obtain,
using \eqref{eq:est-C-m-sum} and $C_\delta\ge 1$,
\begin{align}
  \prod_{i=x+1}^{j-1}\!\!\frac{1}{1-\frac{m}{W_{i-1,i}}}
  \cdot \frac{1}{1-\frac{m+p_j}{W_{j-1,j}}}
  \le 2\!\!\prod_{i=x+1}^{j-1}\!\! e^{2\frac{m}{W_{i-1,i}}}
  \le 2e^{2m\sum_{i=x+1}^{y-1}\frac{1}{W_{i-1,i}}}
    \le 2e^{2\kappa/3}\le 2e.
\end{align}
Substituting this and \eqref{eq:est-1-over-1-C-delta} in the definition
\eqref{eq:def-Y-e-j} of $Y_{e,j}$, it follows, using in addition
$\log(1+\delta_j)\ge\delta_j/2$ since $\delta_j=j^{-(\alpha-\gamma-1)}\in[0,1]$,
\begin{align}
  Y_{e,j}& \le 4e\frac{2^{(j-x)m}}{(1+\delta_j)^{p_j}}
           =4e\cdot e^{(j-x)m\log 2-p_j\log(1+\delta_j)}
           \le 4e\cdot e^{(j-x)m\log 2-\frac12p_j\delta_j}\nonumber\\
&\le  4e\cdot e^{jm_x\log 2-\frac12p_j\delta_j}
  =4e\cdot e^{\cvier\overline W x^\gamma j
    -2\cvier\overline W j^{\gamma+1}}
  \le 4e\cdot e^{-\cvier\overline W j^{\gamma+1}}
  \le 4e\cdot e^{-\cvier\overline W j},
  \label{eq:bound-Y-e-j}
\end{align}
where we used $x\le j$ and $\gamma\ge 0$ in the last two inequalities, respectively.

\paragraph{Explanation for the choice of parameters.} In \eqref{eq:bound-Y-e-j}, we have a competition between the numerator $2^{(j-x)m}$, which is nondecreasing in $j$,
and the denominator $(1+\delta_j)^{p_j}$ with sufficiently small $\delta_j$ to ensure
summability of $\sum_j\sqrt{\delta_j}$. Recall that this denominator arose from the
Chebyshev-type bound \eqref{eq:chebyshev-type-bound-chi}. 
On the one hand, the parameter $p_j$ has to be
chosen large enough for the denominator to win this competition. On the other hand,
it must not be too large since we need $m_e+p_j<W_{j-1,j}$, cf.\ \eqref{eq:est-m-p-Winv}. 

\paragraph{Proof completion.}
Summing the bound \eqref{eq:bound-Y-e-j} over $j=x+1,\ldots,y$ and adding the bound
\eqref{eq:est-Y-e-x} for $Y_{e,x}$ yields 
\begin{align}
  \label{eq:est-sum-Y-1+rho}
  \sum_{j=x}^yY_{e,j}
  \le 1+\frac{\kappa}{2}+4e\sum_{j=x+1}^y e^{-\cvier\overline W j}
  \le 1+\frac{\kappa}{2}+4e\frac{e^{-\cvier\overline W (x+1)}}{1-e^{-\cvier\overline W}}. 
\end{align}
Using $\overline W\ge\wnull$ and our choice \eqref{eq:def-W0} of
  $\wnull$, we observe $e^{-\cvier\overline W}\le\sqrt\kappa/6$. Together
  with $x\ge 1$ and $\kappa\le 1$ this yields 
\begin{align}
  4e\frac{e^{-\cvier\overline W (x+1)}}{1-e^{-\cvier\overline W}}
  \le 4e\frac{e^{-2\cvier\overline W}}{1-e^{-\cvier\overline W}}
  \le 4e\frac{\left(\frac{\sqrt\kappa}{6}\right)^2}{1-\frac16}
  =\frac{2e}{15}\kappa
  \le \frac{\kappa}{2}.
\end{align}
We conclude $\sum_{j=x}^yY_{e,j}\le 1+\kappa$. 
Substituting this bound 
and the definition \eqref{eq:def-Y-e} of $Y_e$ in
\eqref{eq:E-prod-B-prod-Y}, we conclude
\begin{align}
  \label{eq:E-prod-B-prod-Y2}
 & \E\left[\prod_{e\in E'\cup E''}B_{x_ey_e}^{m_e}\right]\le
  \prod_{e\in E'}Y_e\cdot \prod_{e\in E''}\sum_{j=x_e}^{y_e}Y_{e,j}
  \le (1+\kappa)^{|E''|}\prod_{e\in E'}\frac{1}{1-\frac{m_e}{W_e}}.
\end{align}
This proves the first inequality in claim \eqref{eq:main-thm-bound1}.
The second inequality in this claim follows from the bound
$(1-\frac{m_e}{W_e})^{-1}\le 1+\kappa$. This proves part 2 of the theorem.

Finally, part 1 is a special case of part 2. Indeed, \eqref{eq:bound-cosh-ux-uy-main} follows from \eqref{eq:main-thm-bound1}
with $E'=\emptyset$ and $E''=\{e\}$. 
\end{proof}

We conclude this section with the quantitative version of Theorem
\ref{thm:one-dim} announced earlier. 

\begin{lemma}[Quantitative version of Theorem~\ref{thm:one-dim}]
  \label{le:thm-quantitative}
  \mbox{}\\
  In Theorem~\ref{thm:one-dim}, one may choose $\gamma\in[0,\alpha-3)$, 
  \begin{align}
    \label{eq:def-W0}
  \wnull=\wnull(\kappa,\alpha,\gamma)=-\frac{1}{2\cvier}\log \frac{\kappa}{36}, \quad
  \cdrei=  \frac{\cvier}{\log 2} 
\end{align}
with the constant
\begin{align}
  \label{eq:def-czwei}
  \cvier=\cvier(\kappa,\alpha,\gamma)=\min\left\{\frac{2\log 2}
  {3e^{\frac{2\sqrt{2}}{\alpha-\gamma-3}}}\kappa,\frac{1}{16}\right\}
  .
\end{align}
In the case $\gamma=0$ and $\kappa=1$, we have 
\begin{align}
  \label{eq:expression-W-0-special-case}
  \wnull=\max\left\{16\log 6, 
  \frac32 e^{\frac{2\sqrt{2}}{\alpha-3}}\log_26\right\}.
\end{align}
\end{lemma}
\begin{proof}
  In the proof of Theorem~\ref{thm:one-dim}, we showed that $\wnull$, $\cdrei$,
  and $\cvier$ can be chosen as in \eqref{eq:def-W0} and \eqref{eq:def-czwei}.
  For $\kappa=1$, we obtain $\wnull=\log(6)/\cvier$, which yields the expression in 
  \eqref{eq:expression-W-0-special-case}.
\end{proof}

\section{Bounds in the hierarchical model}
\label{sec:bounds-hier-lr}

In this section we use the bounds obtained for the one-dimensional
inhomogeneous model to prove our main result. 

\smallskip\noindent
\begin{proof}[Proof of Theorem~\ref{thm:hierarchical}]
  The first equality in \eqref{eq:hierarchical-est-cosh-ui} follows from the definitions
  \eqref{eq:change-of-coordinates2} and \eqref{eq:def-Sij-Bij} of $z_i$ and $B_{i\rho}$. 
  Since we consider uniform pinning and the hierarchical weights from 
  \eqref{eq:hierarchical-weights}, 
  $(u_i,s_i)$, $i\in\Lambda_N$, are identically distributed. By
  \cite[Corollary 2.3]{disertori-merkl-rolles2020} the law of $(e^{u_{i_1}},s_{i_1}e^{u_{i_1}})$ with
  $i_1:=(0,\ldots,0)\in\{0,1\}^N$ agrees with the law of $(e^{u_{i_1}},s_{i_1}e^{u_{i_1}})$
  in an effective model
  on the complete graph with vertex set $A=\{j,i_1,\ldots,i_N\}$, weights
  $W_{i_{l-1},i_l}=2^{2l-3}w^H(l)$, and pinning strengths
  $h_{i_l}:=2^{l-1}h^H$, $2\le l\le N$, $h_{i_1}:=h^H$.   
  The other weights also need to have specific values in this effective
  model, but they are not displayed here because these values do not play any role
  in the subsequent arguments. As a consequence, the law of $(u_{i_1},s_{i_1})$
  in these two models agrees. This implies 
    \begin{align}\label{eq:hierantichain}
      \E_{W^H,h^H}^{\Lambda_N}[B_{i_1\rho}^m]
      =\E_{W^H,h^H}^{\Lambda_N}\Big[\big(\cosh u_{i_1}+\tfrac12s_{i_1}^2e^{u_{i_1}}\big)^m\Big]
    =\E_{W,h}^A[B_{i_1\rho}^m]. 
  \end{align}
  We regard the pinning point $\rho$ as an additional vertex at the end of the antichain
  $j,i_1,\ldots,i_N$, writing $i_{N+1}:=\rho$. Note that by assumption
    \eqref{eq:ass-weights-pinning}, one has
    $W_{i_{l-1}i_l}=2^{2l-3}w^H(l)\ge\overline W l^\alpha$ for $l\in\{2,\ldots,N\}$
    and $W_{i_N,i_{N+1}}=h_{i_N}=2^{N-1}h^H\ge \overline W (N+1)^\alpha$.
    Hence, we can apply the bound \eqref{eq:bound-cosh-ux-uy-main} in Theorem
    \ref{thm:one-dim} with $x=i_1$ and $y=i_{N+1}$, in case P2, which yields claim \eqref{eq:hierarchical-est-cosh-ui}.

    Recall that $\ceins=\frac23 e^{-\frac{2\sqrt{2}}{\alpha-3}}$. 
    Taking $\gamma=0$ in \eqref{eq:def-W0} and \eqref{eq:def-czwei} and using
    $\kappa\le(16\ceins\log 2)^{-1}$ from \eqref{eq:assumptions-kappa-overline-W-m}
    yields 
   \begin{align}
     \label{eq:ceins-overline-W-with-gamma-equal-zero1}
     \cdrei=&\cdrei(\kappa,\alpha,\gamma=0)=\min\left\{\frac23\kappa e^{-\frac{2\sqrt{2}}{\alpha-3}},\frac{1}{16\log 2}\right\}
              =\ceins\kappa, \\
     \wnull=&\wnull(\kappa,\alpha,\gamma=0)=\frac{1}{2\cdrei}\log_2\frac{36}{\kappa}
              =\frac{1}{2\ceins\kappa}\log_2\frac{36}{\kappa}.
              \label{eq:ceins-overline-W-with-gamma-equal-zero2}
   \end{align}
   Hence, by \eqref{eq:assumptions-kappa-overline-W-m}, the assumptions
   $\overline W\ge\wnull$ and $0\le m\le\cdrei\overline W$
   of Theorem \ref{thm:one-dim} are satisfied. 

   We now treat the two examples in the theorem. For the first example,
   let $s\in(0,1)$ and $\kappa=\overline W^{-s}$. For $\overline W$ large
   enough, depending on $\alpha$ and $s$, one has $\kappa\le(16\ceins\log 2)^{-1}$ and
   $\frac{1}{2\ceins}\frac{1}{\kappa}\log_2\frac{36}{\kappa}
   =\frac{1}{2\ceins}\overline W^s\log_2(36\overline W^s)\le\overline W$.

   For the second example, let $\kappa=\czwei\frac{\log\overline W}{\overline W}$
   with $\czwei>1/(2\ceins\log 2)$.
   For $\overline W$ large enough, depending on $\alpha$,
   it follows $\kappa\le(16\ceins\log 2)^{-1}$ and
   \begin{align}
    \frac{1}{2\ceins}\frac{1}{\kappa}\log_2\frac{36}{\kappa} 
     =\frac{1}{2\ceins\czwei\log 2}\overline W\left(
     1+\frac{\log 36-\log\czwei-\log\log\overline W}{\log\overline W}\right)
     \le\overline W. 
   \end{align}
   Thus, the conditions in \eqref{eq:assumptions-kappa-overline-W-m}
   are satisfied in both cases.   
\end{proof}

\medskip\noindent
\begin{proof}[Proof of Corollary \ref{cor:hierarchical-differences}]
  The first equality follows from \eqref{eq:def-A} and \eqref{eq:def-Sij-Bij}. 
  By \eqref{eq:triangle-strat}, one has $B_{ij}\le 2B_{i\rho}B_{j\rho}$. Applying first
  Cauchy-Schwarz
  and then \eqref{eq:hierarchical-est-cosh-ui}, we deduce \eqref{eq:hierarchical-est-cosh-ui-uj}
  as follows:
  \begin{align}
    \E_{W^H,h^H}^{\Lambda_N}&[B_{ij}^{m/2}]
                              \le 2^{m/2}\E_{W^H,h^H}^{\Lambda_N}[B_{i\rho}^{m/2}B_{j\rho}^{m/2}]\nonumber\\
    &\le 2^{m/2}\E_{W^H,h^H}^{\Lambda_N}[B_{i\rho}^m]^{1/2}\E_{W^H,h^H}^{\Lambda_N}[B_{j\rho}^m]^{1/2}
          \le 2^{m/2}(1+\kappa).
          \label{eq:from-cosh-ui-to-ui-minus-uj}
  \end{align}
\end{proof}

\smallskip\noindent
\textbf{Acknowledgement.}
We would like to thank an anonymous referee for very carefully reading our
manuscript and providing many constructive suggestions for improvements. 

This work was supported by the Deutsche Forschungsgemeinschaft (DFG, German Research Foundation) priority program SPP 2265 Random Geometric Systems and partly funded by DFG -- Germany's Excellence Strategy - GZ 2047/1, projekt-id 390685813. 
M.D.\ acknowledges the hospitality of the Newton Institute in Cambridge, where part of this work has been carried out.

%\bibliographystyle{alpha}
%\bibliography{mathphys}
\end{document}